\documentclass[a4paper,USenglish,autoref]{lipics-v2021}

\usepackage{mathtools}
\usepackage{tikz-cd}
\usepackage{float}
\usepackage{dsfont}
\usepackage{mathrsfs}

\renewcommand\epsilon{\varepsilon}
\renewcommand\hat{\widehat}

\newcommand\F{\mathbb{F}}
\newcommand\Fq{\mathbb{F}_q}
\newcommand\PP{\mathbb{P}}
\newcommand\Z{\mathbb{Z}}
\newcommand\N{\mathbb{N}}

\newcommand\floor[1]{\left\lfloor #1 \right\rfloor}
\newcommand\ceil[1]{\left\lceil #1 \right\rceil}
\newcommand\gen[1]{\left\langle #1 \right\rangle}
\newcommand\fold[1]{\textsf{\textbf{Fold}}\left[#1\right]}
\newcommand\poles[1]{\div_\infty\left( #1 \right)}
\newcommand\zeroes[1]{\div_0\left( #1 \right)}

\renewcommand{\hat}{\widehat}
\renewcommand{\bar}{\overline}

\DeclareMathOperator{\polylog}{polylog}
\DeclareMathOperator{\Supp}{Supp}
\DeclareMathOperator{\Aut}{Aut}
\DeclareMathOperator{\Div}{Div}
\let\div\relax\DeclareMathOperator{\div}{div}

\def\rest#1#2{{{#1}_{|#2}}}

\newcommand\RS[1]{\mathsf{RS}\left[\F,#1\right]}

\newcommand\prover{\mathsf{P}}
\newcommand\verifier{\mathsf{V}}

\renewcommand{\vec}[1]{\boldsymbol{#1}}
\newcommand{\set}[1]{\left\{#1\right\}}
\newcommand{\size}[1]{\left|#1\right|}
\newcommand{\range}[1]{\set{0,\dots,#1-1}}
\newcommand{\Range}[2]{\set{#1,\dots,#2}}

\newcommand\bfz{{\vec{z}}}

\newcommand\calY{\mathcal{Y}}
\newcommand\calX{\mathcal{X}}

\newcommand\calH{\mathcal{H}}
\newcommand\calF{\mathcal{F}}
\newcommand\calP{\mathcal{P}}

\newcommand\calT{\mathcal{T}}

\newcommand\calG{\mathcal{G}}

\newcommand\seqG{\vec{\calG}}
\newcommand\XGseq{(\calX, \vec{\calG})}

\newcommand\xpi[1]{#1^{(i)}}
\newcommand\xpip[1]{#1^{(i+1)}}

\newcommand{\err}{\textsf{err}}

\newcommand\mydef{\coloneqq}

\newcommand{\cd}{\cdot}
\newcommand{\ii}{i_{\rm{max}}}

\newtheorem{notation}{Notation}

\title{Interactive Oracle Proofs of Proximity to Algebraic Geometry Codes}

\author{Sarah {Bordage}}%
{LIX, CNRS UMR 7161, Ecole Polytechnique, Institut Polytechnique de Paris \and Inria, Palaiseau, France}%
{sarah.bordage@lix.polytechnique.fr}%
{}%
{Supported by the Chair ``Blockchain \& B2B Platforms'', led by l'\emph{X -- {\'E}cole Polytechnique} and the \emph{Fondation de l'{\'E}cole Polytechnique}, sponsored by \emph{Capgemini}, \emph{NomadicLabs} and \emph{Caisse des Dépôts}.}

\author{Mathieu {Lhotel}}%
{Laboratoire de Mathématiques de Besançon, UMR 6623 CNRS, Université de Bourgogne Franche-Comté, France}%
{mathieu.lhotel@univ-fcomte.fr}%
{}
{}

\author{Jade {Nardi}}%
{Univ Rennes, CNRS, IRMAR - UMR 6625, F-35000 Rennes, France}%
{jade.nardi@univ-rennes1.fr}%
{https://orcid.org/0000-0003-0901-7266}
{Suported by the French government ``Investissements d’Avenir'' program ANR-11-LABX-0020-01.}

\author{Hugues {Randriam}}%
{ANSSI, Paris, France \and Institut Polytechnique de Paris, Télécom Paris, Palaiseau, France}%
{randriam@telecom-paris.fr}%
{}
{}

\authorrunning{S. Bordage, M. Lhotel, J. Nardi and H. Randriam}

\Copyright{Sarah Bordage, Mathieu Lhotel, Jade Nardi and Hugues Randriam}



\category{} 

\relatedversion{} 

\funding{This work was funded in part by the grant ANR-21-CE39-0009 - BARRACUDA from the French National Research Agency.}

\acknowledgements{The authors are very appreciative to the anonymous reviewers whose comments and suggestions helped to improve and clarify this manuscript.
    The third author thanks Marc Perret for his precious advice in the early days of this project. The authors are grateful to Daniel Augot for suggesting to work on this project and for many valuable discussions. They also thank Eli Ben-Sasson and Alessandro Chiesa for their helpful insights.}

\EventEditors{John Q. Open and Joan R. Access}
\EventNoEds{2}
\EventLongTitle{42nd Conference on Very Important Topics (CVIT 2016)}
\EventShortTitle{CVIT 2016}
\EventAcronym{CVIT}
\EventYear{2016}
\EventDate{December 24--27, 2016}
\EventLocation{Little Whinging, United Kingdom}
\EventLogo{}
\SeriesVolume{42}
\ArticleNo{23}

\begin{document}
	\nolinenumbers

\maketitle

\begin{abstract}
In this work, we initiate the study of proximity testing to Algebraic Geometry (AG) codes. An AG code $C = C(\calX, \calP, D)$ over an algebraic curve $\calX$ is a vector space associated to evaluations on $\calP \subseteq \calX$ of functions in the Riemann-Roch space $L_\calX(D)$. The problem of testing proximity to an error-correcting code $C$ consists in distinguishing between the case where an input word, given as an oracle, belongs to $C$ and the one where it is far from every codeword of $C$. AG codes are good candidates to construct probabilistic proof systems, but there exists no efficient proximity tests for them. We aim to fill this gap.

We construct an Interactive Oracle Proof of Proximity (IOPP) for some families of AG codes by generalizing an IOPP for Reed-Solomon codes, known as the \textsf{FRI} protocol \cite{BBHR18a}. We identify suitable requirements for designing efficient IOPP systems for AG codes.    
Our approach relies on a neat decomposition of the Riemann-Roch space of any invariant divisor under a group action on a curve into several explicit Riemann-Roch spaces on the quotient curve. We provide sufficient conditions on an AG code $C$ that allow to reduce a proximity testing problem for $C$ to a membership problem for a significantly smaller code $C'$.

As concrete instantiations, we study AG codes on Kummer curves and curves in the Hermitian tower. The latter can be defined over polylogarithmic-size alphabet. We specialize the generic AG-IOPP construction to reach linear prover running time and logarithmic verification on Kummer curves, and quasilinear prover time with polylogarithmic verification on the Hermitian tower.

\end{abstract}


\section{Introduction}
Let $C \subset \Sigma^S$ be an evaluation code with evaluation domain $S$ of size $n$ and alphabet $\Sigma$.
For $u \in \Sigma^S$, if $\Delta(u, C) > \delta$, we say that $u$ is $\delta$-far from $C$ and $\delta$-close otherwise. We address the problem of proximity testing to a code $C$, \textit{i.e.} given a code $C$ and assuming a verifier has oracle access to a function $f : S \rightarrow \Sigma$, distinguish between the case where $f \in C$ and $f$ is $\delta$-far from $C$. In this paper, we focus on the case where $C$ is an AG code. An algebraic geometry (AG) code $C = C(\calX, \calP, D)$ is a vector space formed by evaluations on a set $\calP\subset \calX$ of functions in the Riemann-Roch space $L_\calX(D)$. We address this problem in the Interactive Oracle Proof model \cite{BCS16}, which has demonstrated to be particularly promising for the design of proof systems in the past few years.

\subparagraph*{Context of this work.}
Under the generic term of \emph{arithmetization} \cite{LFKN90}, algebraic techniques for constructing proof systems using properties of low-degree polynomials have emerged from the study of interactive proofs \cite{Bab85,GMR85}. Arithmetization techniques have been enhanced and fruitfully applied to other broad families of proof systems since then, including probabilistically checkable proofs (PCPs, \cite{BFLS91, AS92, ALM+98}). Loosely speaking, in a probabilistic proof system for a binary relation $\mathcal R$, the arithmetization process transforms any instance-witness pair $(x, w)$ into a word that belongs to a certain error-correcting code $C$ if $(x, w) \in \mathcal R$, and is very far from $C$ otherwise.

Since the seminal works of Kilian \cite{Kil92} and Micali \cite{Mic95}, a lot of efforts have been put into making PCPs efficient enough to obtain \emph{practical} sublinear non-interactive arguments for delegating computation. In search of reducing the work required to generate such probabilistic proofs, as well as the communication complexity of succinct arguments based on them, Interactive Oracle Proofs (IOPs, \cite{BCS16,RRR16}) have been introduced as a common generalization of PCPs, IPs and IPCPs \cite{KR08}

Considering for the first time univariate polynomials instead of multivariate ones, \cite{BS08, Dinur07} constructed a PCP with quasilinear proof length and constant query complexity. Since then, efficient transparent and zero-knowledge non-interactive arguments have been designed by relying on Reed-Solomon (RS) codes, including \cite{AHIV17,BBHR19,BCRSVW19,BCGGRS19,KPV19,COS20,ZXZS20}. At some point, aforementioned sublinear arguments require a proximity test for RS codes. 

As a solution, one can use an IOP of Proximity for Reed-Solomon codes (an IOP of Proximity \cite{BCGRS17} is the natural extension to the IOP model of a PCP of Proximity). In an IOP of Proximity (IOPP) for an error-correcting code, a verifier is given as input a code, has oracle access to a function (a purported codeword) and interacts with a prover. After the interaction, the verifier accepts if the function is indeed a codeword, and rejects with high probability if the function is far from any codeword. We defer the formal definition of an IOPP to Section \ref{subsec:def-iopp}.

The \textsf{FRI} protocol is a prover-efficient IOP of Proximity for testing proximity to Reed-Solomon codes evaluated over well-chosen evaluation points (introduced by \cite{BBHR18a} and further improved in \cite{BKS18}, \cite{BGKS20}, \cite{BCIKS20}). It admits linear prover time, logarithmic verifier time and logarithmic query complexity. While being sub-optimal for some parameters, the \textsf{FRI} protocol is highly-efficient in practice and is a crucial tool in systems deployed in the real-world. For instance, \cite{BCGRS17,RR20} proposed IOPP constructions for RS codes with constant query complexity whereas the \textsf{FRI} protocol has logarithmic query complexity. 

The main drawback of RS codes is that they must have an alphabet size larger than their length. AG codes \cite{Goppa77}, as evaluations of a set of functions at some designated rational points on a given curve, extend the notion of Reed-Solomon codes and inherit many of their interesting properties. Therefore, replacing RS codes with AG codes is not only natural but has also led to improvements in the past. Examples of cryptographic applications of AG codes include public key cryptography, distributed storage, secret sharing and multi-party computation. A feature for a family of codes that facilitates arithmetization is a multiplication property \cite{Mei13}, namely the fact that the component-wise multiplication of two codewords results in codewords in a code whose minimum distance is still good. This multiplication property actually emulates multiplication of low-degree polynomials. Algebraic geometry codes not only feature this multiplication property but may also have arbitrary large length given a fixed finite field $\F$, unlike RS codes.

\subparagraph*{Limitations of Reed-Solomon codes.}
We identify two limitations of using RS codes in IOPs.

As mentioned earlier, RS codes are the simplest case of AG codes, but possess an inherent limitation: the alphabet size must be larger than the block length of the code. Therefore, practical IOP-based succinct arguments are designed over \emph{large fields}.

The second limitation is related to the algebraic structure of the field. RS-IOPPs \cite{BS08, BBHR18a} require the set $D \subset \F$ of evaluation points to have a special structure. Concretely,
the field must contain a subgroup of \emph{large smooth order}, typically a power of 2 which is larger than the size of the non-deterministic computation to be verified. Depending on the applications of succinct non-interactive arguments, a base field might already be imposed. This is for instance the case for standard digital signature schemes. For computations whose size exceeds the order of the largest smooth subgroup of the field, RS-IOPPs known to date can no longer be used.    

We observe that lowering the size of field elements may not significantly shorten the length of IOP-based succinct non-interactive arguments (see \cite{BCS16}). There are, however, other reasonable motivations to replace RS codes with AG ones. We explain below how AG codes could circumvent the limitations of RS codes.

\subparagraph*{Why can AG codes be useful?}
First, working over smaller fields lowers the cost of field operations\footnote{Consider the application of checking the correct execution of a size $n$ computation. Then an RS-based IOP for this problem will work over a field of size $\Omega(n)$. This means that a single addition of two field elements will cost $\Omega(\log n)$ operations. If the IOP is instead based on a code with polylogarithmic-size alphabet, the cost of a single addition is only $\Omega(\log(\polylog(n))$.}. For concrete efficiency, complexity measures such as prover time and verifier time are closely examined. Reducing significantly the size of the alphabet would have a direct impact on the binary cost of arithmetic operations.   
Smaller fields could enhance efficiency of proof systems since arithmetization of general circuits would be more efficient. Moreover, on the prover side, the bit complexity of encoding codewords might be smaller.

A popular belief is that encoding with AG codes is an heavy task. It is surely true in general, but there are explicit families of AG codes for which there are quasilinear time encoding algorithms \cite{BRS20}. We discuss more about encoding in Section \ref{subsec:ag-encoding}. On another note, putting forward applications of AG codes can motivate the study of fast encoding algorithms for AG codes, in particular in the computer algebra community.

One may be concerned by the overhead of reducing the alphabet size when targeting a soundness error less than $2^{-\kappa}$. Notice that it is possible to sample enough bits of randomness from an extension field when needed, or to repeat only some parts of the protocol (see \cite{ethSTARK,BBHR18a}). For instance, the soundness error of our IOPP is bounded from below by $\size \F^{-1}$. Reaching the targeting soundness requires to repeat the interactive phase of the IOPP $s$ times, inducing a factor $s \simeq \frac{\kappa}{\log \size \F}$ multiplicative overhead for the prover. A rough estimation of bit complexities does not show evidence of a significant overhead. Overall, a proof system supporting small fields might be more efficient: any part of the protocol which does not contribute to the soundness error could benefit from cheaper field operations.

In addition, AG codes offer more flexibility on the choice of the field. For computations of size $n$, we propose AG codes for which the field is not required to admit an $n$-th root of unity (unlike RS-IOPP on a prime field). Specifically, for AG codes over Kummer curves, the base field needs only to have a $N$-th root of unity, where $N$ divides $n$. For AG codes over curves in a Hermitian tower (which admits a polylogarithmic-size alphabet), \emph{our IOPP does not involve any assumption on the alphabet}, except that it must be a degree-2 extension of a field $\Fq$, where $q$ is any prime power.

Finally, the question of whether there exist concretely efficient IOPPs for AG codes is motivated by both a theoretical and practical perspective. 

\subsection{Definition of an IOPP for a code}\label{subsec:def-iopp}

We are specifically interested in public-coin IOP of Proximity (IOPP) for a family of evaluation codes $\mathscr{C}$, thereby we specify our definition for this particular setting. An IOPP  $(\prover, \verifier)$ for a code $C$ is a pair of randomized algorithms, where both $\prover$ (the prover) and $\verifier$ (the verifier) receive as explicit input the specification of a code $C \in \mathscr{C}$, $C \subseteq \Sigma^S$. We define the input size to be $n = \size{S}$. Furthermore, a purported codeword $f : S \rightarrow \Sigma$ is given as explicit input to $\prover$ and as an oracle to $\verifier$. The prover and the verifier interact over at most $\mathsf{r}(n)$ rounds. During this conversation, $\prover$ seeks to convince $\verifier$ that the purported codeword $f$ belongs to the code $C$. 

At each round, the verifier sends a message chosen uniformly and independently at random, and the prover answers with an oracle. Verifier's queries to the prover's messages are generated by public randomness and performed after the end of the interaction with the prover. Thus, such an IOPP is in particular a \emph{public-coin} protocol (or Arthur-Merlin \cite{Bab85}). 

Let us denote $\langle \prover \leftrightarrow \verifier \rangle \in \set{\textsf{accept}, \textsf{reject}}$ the output of $\verifier$ after interacting with $\prover$. The notation $\verifier^f$ means that $f$ is given as an oracle input to $\verifier$. We say that a pair of randomized algorithms $(\prover, \verifier)$ is an IOPP system for the code $C \subseteq \Sigma^S$ with \emph{soundness error} $s : (0,1] \rightarrow [0,1]$, if the following conditions hold:

\begin{description}
	\item[Perfect completeness:] If $f \in C$, then $\Pr[\langle \prover (C, f)\leftrightarrow \verifier^{f}(C) \rangle = \textsf{accept}] = 1$.
	\item[Soundness:] For any function $f \in \Sigma^S$ such that $\delta \mydef \Delta(f, C) > 0$ and any unbounded malicious prover $ \prover^*$, $\Pr[\langle \prover^* \leftrightarrow \verifier^{f}(C) \rangle = \textsf{accept}] \leq s(\delta).$
\end{description}

The length of any prover message is expressed in the number of symbols of an alphabet $\mathsf{a}(n)$. The sum of lengths of prover's messages defines the proof length $\mathsf{l}(n)$ of the IOPP. The query complexity $\mathsf{q}(n)$ is the total number of queries made by the verifier to both the purported codeword $f$ and the oracle sent by the prover during the interaction. The prover complexity $t_p(n)$ is the time needed to generate prover messages during the interaction (which does not include the input function $f$). The verifier complexity $t_v(n)$ is the time spent by the verifier to make her decision when queries and query-answers are given as inputs.

\subsection{Our results}
In this section, we provide and overview of the three contributions of this paper. In all this work, we state complexities in field elements and field operations, where the field is the alphabet of the considered code. Asymptotic complexities are relative to the length of the code.

\begin{itemize}
	\item The first one is a clear criterion for constructing IOPPs with linear proof length and sublinear query complexity for AG codes. Our hope with this result is to open up new possibilities for designing efficient probabilistic proof systems based on families of AG codes with constant rate and constant distance.
	\item The second contribution is a concrete instantiation for AG codes defined over Kummer-type curves. This IOPP has strictly linear prover time and strictly logarithmic verification (counted in field operations). Thus, we give a strict generalization of the \textsf{FRI} protocol for codes of length $n$ over an alphabet of size roughly $n^{2/3}$.
	\item The third one is a concrete instantiation for AG codes defined over a tower of Hermitian curves. Considering recursive towers enables to construct an IOPP for AG codes with \emph{polylogarithmic-size} alphabet. For those codes, we give an IOPP with quasilinear prover time and polylogarithmic verification (counted in field operations).
\end{itemize}

Efficiency of our two AG-IOPP instantiations leverages the fact that proximity testing for these families of AG code can be reduced to a proximity test for a small RS code.

\paragraph*{Generic criterion for constructing AG-IOPPs}
Let $\calX$ be a curve defined over a finite field $\F$, $D$ a divisor on the curve $\calX$ and $\calP \subset \calX(\F)$. This defines an AG code $C=C(\calX, \calP, D)$. We construct a sequence of curves
\begin{center}
    \begin{tikzcd}
    \calX \mydef \calX_0 \arrow[r, "\pi_0"] & \calX_1 \arrow[r, "\pi_1"] & \calX_2 \arrow[r, "\pi_2"] & \cdots \arrow[r, "\pi_{r-1}"] & \calX_r,
    \end{tikzcd}
\end{center}
so that $\calX_{i+1}$ arises as the quotient of the curve $\calX_i$ by some automorphism subgroup $\Gamma_i$ under the quotient map $\pi_i$. 

Using these consecutive projection maps, we construct a sequence of AG codes $(C_i)_{0 \le i \le r}$ of decreasing length to turn the proximity test of the function $f^{(0)}=f$ to $C_0$ into a membership test of a function $f^{(r)}$ in $C_r$. We show that such a procedure is possible if there is a large enough (with respect to the length of the code $C_0$) group $\calG$ in the automorphism group of $\calX$ and under some hypotheses on the divisor $D$ (overviewed in Section \ref{subsec:AG-IOPP-hyp}, detailed in Section \ref{sec:foldable}). An AG code fulfilling all the required conditions is called \emph{foldable}.

Assuming that an AG code $C(\calX, \calP, D)$ of blocklength $n$ is foldable, we show that there is an $O(\log n)$-round IOPP for it, with linear proof length, sublinear query complexity and constant soundness (see Theorem \ref{thm:properties}).

In general, we observe that the larger is the group $\calG$ acting on $\calX_0$ compared to $n$, the smaller are the query complexity and the verifier decision complexity of the protocol. 

However, we notice that the hypothesis on the size of $\calG$ is not a necessary condition for constructing an IOPP with sublinear verification. For instance, if the curve $\calX_r$ is isomorphic to the projective line $\PP^1$, we can continue to recurse in order to reduce even more the size of the proximity testing problem. We propose two interesting families of AG codes for which it is the case.

\paragraph*{Concrete IOPP for AG codes on Kummer curves}
When $\calX$ is a Kummer curve of the form $y^N = f(x)$, we show how to choose $\calP$ and $D$ to make the AG code $C = C(\calX, \calP, D)$ foldable. We benefit from the action of the group $\Z / N \Z$ on $\calX$ that yields a quotient curve $\calX / (\Z/N\Z)$ isomorphic to the \emph{projective line}. This enables us to define a sequence of codes $(C_i)_{0 \le i \le s}$ such that $C_0 = C$ and the code $C_s$ is a \textit{Reed-Solomon code} of dimension $(\deg{D})/N + 1$.

\begin{theorem}[Informal, see Theorem \ref{thm:kummer-properties}]\label{thm:intro-kummer} Let ${C = C(\calX, \calP, D)} \subset \F^{\calP}$ be a foldable AG code defined over a Kummer curve $\calX$  of equation $\calX : y^N = f(x)$ such that $\deg f= N \ell -1$ for some integer $\ell >0$ and $N$ is a smooth integer, coprime with $\size{\F}$.    Assume $\F$ contains a primitive $N$-th root of unity. 
	The block length $n \mydef \size{\calP}$ is a multiple of $N$ and satisfies $n < \ell N^2\size{\F}^{1/2}$.
	Let $f : \calP \rightarrow \F$ be a purported codeword. For every proximity parameter $\delta \in (0, 1)$ and soundness $\epsilon \in (0,1)$, there exists a public-coin IOPP system $(\prover, \verifier)$ for $C$ with perfect completeness and the following properties:
	
	$
	\begin{array}{lll}
		\text{rounds} &\mathsf{r}(n) &< \log n,\\
		\text{proof length} &\mathsf{l}(n) &= O(n),\\
		\text{query complexity} &\mathsf{q}(n) &= O(\log n),\\
		\text{prover complexity} &\mathsf{t_p}(n) &= O(n),\\
		\text{verifier decision complexity} & \mathsf{t_v}(n) &= O(\log n).
	\end{array}
	$
\end{theorem}

It is worth noting that the Hermitian curve defined over $\F_{q^2}$ by $y^{q+1}=x^q+x$ satisfies the hypotheses of the previous theorem. It is well known to be \emph{maximal}, i.e. it has the maximum number of rational points with respect to its geometry. Besides, we recall that Hermitian codes over alphabet $\Sigma$ support block length up to $\size \Sigma^{3/2}$.

\paragraph*{Concrete IOPP for AG codes on towers of Hermitian curves}
We recall that a tower of curves consists of an infinite sequence of curves 
\[\PP^1 = \calX_0\gets\calX_1\gets\ \ldots\gets\calX_n\gets\ldots\]
such that the number of rational points of the $n^{th}$ curve tends to infinity as $n$ tends to infinity. Towers of curves play a prominent role in the history of AG codes as they define codes with outstanding length and correction capacity \cite{TVZ82,BBGS14}. 

The Hermitian tower is an example of the widely studied Artin-Scheier extensions \cite{L92,S08}. In this case, the curve $\calX_i$ arises as the quotient of the curve $\calX_{i+1}$ above modulo the action of a group of order $q$ of the finite field $\F_{q^2}$, the curve $\calX_0$ being isomorphic to the projective $\PP^1$. Therefore, one can test proximity to an AG code from one of the curves $\calX_n$ by going down along the tower and then testing proximity to a RS code, whose degree can be expressed explicitly in terms of the initial AG code. 

Beyond supporting polylogarithmic-size alphabet, AG codes over the Hermitian tower happen to be more naturally ``foldable''. In particular, no additional assumptions on the alphabet are required.

We write $\polylog(n)$ for functions that are in $O\left(\log^k(n)\right)$ for some $k$.

\begin{theorem}[Informal, see Theorem \ref{thm:tower-properties}]\label{thm:intro-tower} 
	Let ${C = C(\calX, \calP, D)} \subset \F^{\calP}$ be a foldable AG code over an alphabet $\F$ of size $\size \F = \Omega(\log^k(n))$ for some constant $k$.
	We denote $n = \size{\calP}$. Let $f : \calP \rightarrow \F$ be a purported codeword. For every proximity parameter $\delta \in (0, 1)$ and soundness $\epsilon \in (0,1)$, there exists a public-coin IOPP system $(\prover, \verifier)$ for $C$ with perfect completeness and the following properties:
	
	$
	\begin{array}{lll}
		\text{rounds} &\mathsf{r}(n) &< \log n,\\
		\text{proof length} &\mathsf{l}(n) &= O(n),\\
		\text{query complexity} &\mathsf{q}(n) &= \polylog(n),\\
		\text{prover complexity} &\mathsf{t_p}(n) &= \widetilde{O}(n),\\
		\text{verifier decision complexity} & \mathsf{t_v}(n) &= \polylog(n).
	\end{array}
	$
\end{theorem}

\subsection{More on the practicality of AG codes}\label{subsec:ag-encoding}
When constructing a proximity test for a code, it is assumed that the purported codeword is given as input to the prover. Thus, the prover complexity is computed thereof. While we heavily rely on the group of automorphisms of the curve for proving the existence of an efficient IOPP for ``foldable'' AG codes, we emphasize that the work of the prover and the verifier during the protocol is essentially to perform some univariate polynomial interpolation tasks, with very small degree. In particular, neither the prover nor the verifier of the IOPP system needs to run an encoding algorithm for AG codes.

However, keeping applications to code-based IOP constructions in mind, the running time of the IOP prover is bounded from below by the time needed to encode codewords during arithmetization. 
Fast encoding for AG codes is not the most widely studied computational task, and is often a concern when suggesting constructions based on AG codes\footnote{In general, the asymptotic cost of the task of encoding an arbitrary linear code of length $n$ is $O(n^2)$ (using a generator matrix for the code).}.

This is a reason why we focus our study on families of AG codes that are particularly likely to lead to practical implementations, as we argue next. Specifically, our study includes the two following subfamilies of one-point AG codes over small alphabets with constant rate and distance.

\begin{itemize}
	\item The first family includes one-point AG codes over Kummer-type curves, and in particular the notorious Hermitian curve. \cite{BRS20} proposed an encoding algorithm with quasilinear complexity $\widetilde{O}(n)$. Roughly speaking, \cite{BRS20} method consists in translating the encoding task into a bivariate polynomial multipoint-evaluation problem. Assuming that the evaluation points are well-structured, they view a bivariate polynomial in $\F[X, Y]$ as a polynomial in $\F[X][Y]$ in order to evaluate it thanks to two univariate multipoint evaluations. It is the same idea as the one for computing $m$-dimensional FFT from $m$ (univariate) FFTs.
	\item The second family of one-point AG codes arises from curves on the Hermitian tower and has an alphabet size polylogarithmic in the block length of the code. It is very likely that those codes could also be encoded in quasilinear time, by iteratively applying the encoding method proposed by \cite{BRS20}.
\end{itemize}

We also point out that bases for Riemann-Roch spaces related to these codes are explicitly known.

\subsection{Related work}
We discuss works related to AG-based proximity testing. We emphasize that the motivation behind existing works was only theoretical. In particular, the PCP techniques being used are too complex to be implemented for verifying meaningful computations.

In 2013, \cite{BKKMS13} constructed a PCP with linear proof length and sublinear query complexity for boolean circuit satisfiability by relying on AG codes. More precisely, for any $\epsilon > 0$ and instances of size $n$, their PCP has length $2^{O(1/\epsilon)}n$ and query complexity $n^\epsilon$. When aiming at optimal proof length and query complexity as small as possible, this result remains the state-of-the-art PCP construction. By using AG codes, the authors of \cite{BKKMS13} reduced the field size to a constant, which avoids a logarithmic blowup in proof bit-length (occurring e.g. in \cite{BS08} when using univariate polynomials of degree $m$ to encode binary strings of length $m$). 
In \cite{BKKMS13}, the authors pointed out that they are not able to apply proof composition \cite{AS92} to reduce the query complexity of their PCP because the decision complexity of the PCP verifier is too large (polynomial in the query complexity). 

Improving on \cite{BKKMS13}, \cite{BCGRS17} proposed an IOP for boolean circuit satisfiability with linear proof length and constant query complexity. The IOP of \cite{BCGRS17} invoked the sumcheck protocol \cite{LFKN90} on $O(1)$-wise tensor product of AG codes, which exponentially deteriorates the rate of the base code. Then, they use Mie's PCP of Proximity for non-deterministic languages \cite{Mie09} to test proximity to the tensored code. Both constructions benefit from the use of AG codes to get constant size alphabet and linear proof bit-lengths. 

A recent work of \cite{RR20} constructed an IOPP for any deterministic language which can be decided in time $\text{poly}(n)$ and space $n^{o(1)}$. In particular, \cite[Corollary~3.6]{RR20} can be applied to test proximity to AG codes. This IOPP outperforms our construction on some parameters: it has constant round and query complexities, and proof length is slightly less than $n$. However, it is unlikely that \cite{RR20}'s IOPP leads to a concrete implementation, which is a motivation for our work. Indeed, prover running time is polynomial, and the inner IOPP used for achieving constant query complexity via proof composition is the heavy PCPP of \cite{Mie09}. Mie's PCPP is a theoretical and complex tool used to achieve constant query (e.g in \cite{BCGRS17,BCGGRS19,BCL20}), but it is seen as impractical\footnote{Proposing an alternative to \cite{Mie09} which does not involve heavy PCP machinery would allow to narrow the gap between the best constructions known in theory and the most efficient ones used in practice.}.

By contrast, we exhibit two explicit families of AG codes for which we are able to construct a proximity test with linear prover running time and logarithmic verification (for the first one) and quasilinear prover time with polylogarithmic verification (for the second one). The main point however, is that our construction is undoubtedly much simpler to implement: the most complex task of the prover and the verifier is simply to perform univariate interpolations (with very small degrees). The technical difficulties are in analyzing the conditions allowing the construction, but the protocol itself is very similar to the \textsf{FRI} protocol. IOPP inspired by the \textsf{FRI} protocol have the inherent barrier of logarithmic query complexity. However, in practice, it is still the most efficient proximity test for Reed-Solomon codes known to date.

\section{Technical Overview}

Our IOPP construction relies on the generalization of the \textsf{FRI} protocol to AG codes. Let us first recall some ideas behind the construction of \textsf{FRI} protocol (see e.g. \cite{BKS18} for a detailed presentation). Then we shall describe how we tailor these ideas and which difficulties arise to construct our IOPP for AG codes over curves of positive genus.

\subsection{The \textsf{FRI} protocol for RS proximity testing}

Let $k$ be a positive integer and $\rho \in (0,1)$ such that $\rho = 2^{-k}$. The \textsf{FRI} protocol allows to check proximity to the Reed-Solomon code \[\RS{\calP,\rho} \mydef \set{f \in \F^\calP \mid \deg f < \rho \cd \size{\calP}}\] by testing proximity to $\RS{\calP',\rho}$ with $\size{\calP'}<\size{\calP}$. The \textsf{FRI} protocol considers a family of linear maps $\F^{\calP} \rightarrow \F^{\calP'}$ which randomly ``fold'' any function in $\F^{\calP}$ into a function in $\F^{\calP'}$. We present in a simplified way three key ingredients that enable the \textsf{FRI} protocol to work.

\begin{enumerate}[a)]
	\item\label{it:split} \emph{Splitting of polynomials.} For any polynomial $f$ of degree $\deg f < \rho n$, there exist two polynomials $g,  h$ of degree $< \frac{1}{2} \rho n$ such that 
	\begin{equation}\label{eq:intro-decomp}
		f(x) = g\left(x^2\right) + x \cdot h\left(x^2\right).
	\end{equation}
	One may view such a decomposition as the result of the splitting of the space of polynomials of degree less than $\rho n$ into two copies of the space of polynomials of degree less than $\rho n/2$. 
	
	\item\label{it:fold} \emph{Randomized folding.}
	Choose $\calP$ to be a multiplicative group of order $2^r$ generated by $\omega \in \F$. Then, define $\calP' = \langle \omega^2 \rangle = \{x^2 \mid x \in \calP\}$. Set $\pi : \F \rightarrow \F$ to be the map defined by $\pi(x) = x^2$, observe that $\pi(\calP) = \calP'$. Moreover, $\size{\calP'} = \size{\calP}/2$. The structure of the evaluation domain will allow to reduce the problem of proximity to one of half the size at each round of interaction.
	
	Based on the decomposition \eqref{eq:intro-decomp}, define a \emph{folding operator} $\fold{\cdot,z} : \F^\calP \rightarrow \F^{\calP'}$ for any $z \in \F$ as follows:
	\[\fold{f,z} \mydef g + zh.\]
	If $\deg f < \rho n$, both functions $g : \calP' \rightarrow \F$ and $h : \calP' \rightarrow \F$ belong to $\RS{\calP', \rho}$. Then for any random challenge $z\in \Fq$, the operator $\fold{\cdot, z}$ maps $\RS{\calP,\rho}$ into $\RS{\calP',\rho}$.
	
	\item\label{it:dist} \emph{Distance preservation after folding.} Except with small probability over $z$, we have that if $\Delta(f, \RS{\calP, \rho}) \ge \delta$, then \[\Delta\left(\fold{f, z}, \RS{\calP', \rho}\right) \ge (1 - o(1))\delta.\]
\end{enumerate}

The protocol then goes as follows: the verifier sends a random challenge $z \in \F$ and the prover answers with an oracle function $f' : \calP' \rightarrow \F$, which is expected to be equal to $\fold{f,z} : \calP' \rightarrow \F$. At the next round, $f'$ becomes the function to be ``folded'', and the process is repeated for $r$ rounds. Each round reduces the problem by half, eventually leading to a function $f^{(r)}$ evaluated over a small enough evaluation domain. This induces a sequence of Reed-Solomon codes of strictly decreasing length. The code rate remains unchanged, and so does the relative minimum distance. The final test consists in testing that $f^{(r)}$ belongs to the last RS code. 

Perfect completeness follows from Item \ref{it:fold}. Prover and verifier efficiencies of the \textsf{FRI} protocol come from the possibility of determining any value of $\fold{f,z}$ at a point $y \in \calP'$ with exactly two values of $f$, namely on the set $\pi^{-1}(\set{y})$. Consequently, a single test of consistency between $f$ and $f'$ requires only two queries to $f$ and one query to $f'$.

Soundness of the protocol relies notably on Item \ref{it:dist}. It is proved using results about distance preservation under random linear combinations, that could be roughly stated as follows:
``\emph{Let $V \subset \Fq^{n}$ be a linear code and $g, h \in \Fq^n$. As long as $\delta$ is small enough, if we have $\Delta(g + z h, V) \le \delta$ for enough values $z \in \Fq$, then both $g$ and $h$ are $\delta'$-close to $V$, where $\delta' = (1 - o(1))\delta$}.'' (see \cite{BBHR18a, BKS18, BGKS20, BCIKS20}). Based on that, one can deduce that if $\fold{f, z} = g + zh$ is $\delta$-close to $V$ for enough values of $z$, then both $g$ and $h$ are $\delta'$-close from $V$. The idea of the proof of Item \ref{it:dist} is to exhibit a codeword which is $\delta$-close from $f$, based on the decomposition of Item \ref{it:split}. 

\begin{remark}\label{rmk:even-RS}
	We point out that Item \ref{it:dist} holds because the functions $g$ and $h$ appearing in the decomposition \eqref{eq:intro-decomp} have \emph{exactly} the same degree. This arises from the crucial fact that the \textsf{FRI} protocol considers only RS code of dimension a power of 2. This means that the RS code is defined by polynomials of degree at most an \emph{odd} bound.
	
	Let us give a glimpse of what happens when $f$ is expected to have degree at most an even integer, say $2d$. The degrees of the functions $g$ and $h$ appearing in the decomposition \eqref{eq:intro-decomp} of $f$ are respectively $\deg g \le d$ and $\deg h \le d-1$. 
	Therefore, if $\deg f \le 2d$, then $g + z h$ corresponds to a polynomial of degree $\le d$. However, knowing that $g + z h$ is a polynomial of degree $\le d$ with high probability on $z$ only tells us that both $g$ and $h$ are of degree $\le d$, which is not enough to deduce that $f$ has degree $\le 2d$ and not $2d + 1$. 
	It is worth noting that words corresponding to a polynomial of degree $2d+1$ are among the \emph{farthest} words from the RS code of degree $\le 2d$.       
	In the univariate case, one can overcome this obstacle by supposing not only $\deg g,\deg h \leq d$ but also $\deg (\nu h) \le d$ for a degree-1 polynomial function $\nu$. This implies that $\deg h < d$, hence $\deg f \leq 2d$.
\end{remark}

\subsection{Our IOPP for AG proximity testing}\label{subsec:AG-IOPP-hyp}

Let $\calX$ be a curve defined over a finite field $\F$ and $C=C(\calX,\calP,D)$ be an AG code. We aim to adapt the three ingredients of the \textsf{FRI} protocol to the AG context.

\subparagraph*{Group actions and Riemann-Roch spaces.}
The splitting of the polynomial $f$ into an even and an odd part in \eqref{eq:intro-decomp} comes from the action of a multiplicative group of order $2$ on the evaluation set $\calP$. This observation is also true with the actual \textsf{FRI} protocol, which sets $\pi$ to be an affine subspace polynomial. This phenomenon is likely to occur in a more general framework.

As soon as a group $\Gamma$ of order $p$ acts on the curve $\calX$, its action naturally extends onto the functions on $\calX$. Let us denote by $\pi$ the canonical projection $\pi : \calX \rightarrow \calX / \Gamma$. If we are able to write the Riemann-Roch space associated to $D$ as following
\begin{equation}\tag{$\star$}\label{eq:intro-Kani}
	f= \sum_{j=0}^{p-1} \mu^j f_j \circ \pi \text{ with } f_j \in L_{\calX / \Gamma}(E_j),
\end{equation}

for some function $\mu$ on the curve $\calX$ and some divisors $E_j$ on the quotient curve that are explicitly expressed in terms of the divisor $D$, then we can mimic the decomposition \eqref{eq:intro-decomp} used in the \textsf{FRI} protocol.

Now assume that no point of $\calP$ is fixed by $\Gamma$, \textit{i.e.} for every $P \in \calP$ and $\gamma \in \Gamma$, $\gamma \cdot P \neq P$. Set $\calP'=\pi(\calP)$.
Polynomial interpolation enables the determination of $f_j(P)$ for any point $P \in \calP'$ with exactly $p$ values of $f$, namely on the set $\pi^{-1}(\set{P})$. This means that the decomposition \eqref{eq:intro-Kani} can be written for any function in $\F^\calP$, not only for elements of $L_\calX(D)$.

\subparagraph*{Folding operators.}
From the decomposition \eqref{eq:intro-Kani} above, we want to define a family of folding operators $\left(\fold{\cdot,z}\right)_{z \in \F}$ from $\F^\calP$ to $\F^{\calP'}$ and a code $C'=C(\calX/\Gamma,\calP',D')$ such that ${\fold{\cdot,z}(C) \subseteq C'}$.

\medskip

In a first approach, one could choose to define the folding operators similarly to the \textsf{FRI} protocol by setting for $z \in \F$, \[\fold{f,z}=\sum_{j=0}^{p-1} z^j f_j\] where the functions $f_j$ come from the decomposition \eqref{eq:intro-Kani} of $f \in \F^\calP$. With this definition, the code $C'$ has to be associated to a divisor $D'$ on $\calX/\Gamma$ such that each Riemann-Roch space $L_{\calX / \Gamma}(E_j)$ can be embedded into $L_{\calX / \Gamma}(D')$.

The best scenario is when the divisor $D$ yields a decomposition of $L_\calX(D)$ as $p$ ``copies'' of the same Riemann-Roch space, as it is the case with Reed-Solomon codes of even dimension.  Unfortunately, to the best of our knowledge, it is unlikely that all divisors $E_j$ involved in the decomposition \eqref{eq:intro-Kani} of $f$ are the same (or even equivalent) \emph{if $\calX$ is not the projective line}. We are then facing an issue analogous to the one described in Remark \ref{rmk:even-RS} on $\PP^1$.

Therefore, such a choice of the folding operators does not guarantee the soundness of our protocol.     We thus aim to adapt the idea at the end of Remark \ref{rmk:even-RS} to the AG setting. We introduce some \emph{balancing} functions $\nu_j$ such that, for every $f_j \in C'$, if the product $\nu_j f_j$ also lies in $C'$, then the function $f_j$ belongs to the desired Riemann-Roch space $L_{\calX / \Gamma}(E_j)$. Defining such a balancing function $\nu_j$ is tantamount to specify its pole order at the points supporting the divisor $D'$. The existence of all the functions $\nu_j$ thus depends on the \emph{Weierstrass semigroup} of these points (see \cite[Section 6.6]{HKT13} for definition) and does not hold for any divisor $D'$. If such balancing functions exist for a divisor $D'$, we say that $D'$ is \emph{compatible} with $D$. Finding a convenient divisor $D'$ compatible with a given divisor $D$ is definitely the trickiest part in defining the folding operators properly.

Once we have a divisor $D'$ that is $D$-compatible, we shall embed additional terms in the folding operators to take account of the balancing functions. We shall use more randomness so as not to double the degree in $z$ to avoid a loss in soundness. For $(z_1,z_2) \in \F^2$, we set
\[\fold{f,(z_1,z_2)}=\sum_{j=0}^{p-1} z_1^j f_j+\sum_{j=0}^{p-1} z_2^{j+1} \nu_j f_j.\]
We prove that $\fold{\cdot,(z_1,z_2)}(C) \subseteq C'$, the function $\fold{f,(z_1,z_2)} \in \F^{\calP'}$ can be locally computed from $p$ values of $f$, and $\fold{\cdot,(z_1,z_2)}$ preserves the distance to the code.

\subparagraph*{Finding a decomposition \eqref{eq:intro-Kani}.}

Such a decomposition exists for a cyclic group $\Gamma=\gen{\gamma}$ whose order is coprime with the characteristic. A result of Kani \cite{Kani}  states that, in this case, there exists a function $\mu$ on $\calX$ satisfying \eqref{eq:intro-Kani} such that $\gamma\cdot \mu = \zeta \mu$ where $\zeta$ is a primitive root of unity of order $\size{\calG}$. Moreover, the divisor $E_{j}$ can be explicitly written in terms of the divisor $D$ and the function $\mu$.

Alternatively, if a basis of $L_\calX(D)$ is explicitly known, we may also be able to exhibit such a decomposition without invoking Kani's theorem for some cyclic group of order divisible by the characteristic. This is exactly the strategy we use to design an IOPP for AG codes along the Hermitian tower.

\subparagraph*{Soundness preservation.}

The soundness of the protocol depends on the relative minimum distance of the codes $C$ and $C'$. Ideally , we would like the rates of the codes $C$ and $C'$ to be roughly equal to prevent the relative minimum distance from dropping. In particular, we need $L_{\calX/\Gamma}(D')$ to be not too large with respect to the components $L_{\calX/\Gamma}(E_j)$.

A natural idea would be to choose $D'$ as the divisor $E_j$ with the largest Riemann-Roch space. However, balancing functions only exists for some well-chosen divisors $D'$, whose degree can be significantly larger than the degree of the divisor $E_{j}$ in \eqref{eq:intro-Kani}. Therefore, the divisors $D$ and $D'$ have to be carefully chosen to prevent the minimum distance from collapsing.

\subparagraph*{Sequence of ``foldable'' AG codes.} With the goal of iterating the folding process in mind, we assume that the base curve $\calX_0 \mydef \calX$ is endowed with a \emph{suitable} acting group $\calG$ that we decompose into smaller groups $\Gamma_0, \Gamma_1, \dots, \Gamma_{r-1}$ to fragment its action and create intermediary quotients
\begin{center}
	\begin{tikzcd}
		\calX_0 \arrow[r, "\pi_0"] & \calX_1 \arrow[r, "\pi_1"] & \calX_2 \arrow[r, "\pi_2"] & \cdots \arrow[r, "\pi_{r-1}"] & \calX_r,
	\end{tikzcd}
\end{center}
where the morphism $\pi_i:\calX_i \rightarrow \calX_{i+1}$ is the quotient map by $\Gamma_i$.
For instance, a sufficient condition on the group $\calG$ to have such a sequence is the \emph{solvability}.

A code $C = C(\calX, \calP, D)$ is said to be a \emph{foldable AG code} (Definition \ref{def:good_properties}) if we are able to construct a sequence of AG codes $C_i\mydef C(\calX_i,\calP_i,D_i)$ that support a family of randomized folding operators $\fold{\cdot, \bfz} : \F^{\calP_i} \rightarrow \F^{\calP_{i+1}}$ with the desirable properties for our IOPP (i.e. $\fold{\cdot, \bfz}(C_i) = (C_{i+1})$, local computability, distance preservation to the code). Moreover, to ensure that the last code $C_r$ has sufficiently small length and to obtain an IOPP with sublinear query complexity, we require the size of $\calG$ to be greater than $\size{\calP}^{e}$ for a certain $e \in (0,1)$. Details are provided in Section \ref{sec:foldable}.

\section{Preliminaries}\label{sec:preli}

We start with some reminders on important terms and notations related to the theory of AG codes. We refer readers to \cite{TNV07,S08} for further details on these notions. We will always use $\F$ to denote a finite field.

\subsection{Functions and divisors on algebraic curves}
Let $\calX$ be an algebraic curve defined over a field $\F$. Let $\bar{\F}$ be an algebraic closure of the field $\F$. We denote by $\calX(\F)$ the set of its $\F$-rational points and $\Aut(\calX)$ its automorphism group.

A \emph{divisor} $D$ on $\calX$ is a formal sum of points $D=\sum n_P P$. We say that the divisor $D$ is \emph{effective} if $n_P \geq 0$ for every point $P$. The \emph{support} of $D$, denoted by $\Supp(D)$, is the set of points $P$ for which the coefficient $n_P$ is non zero. We will always consider \textit{rational} divisors, whose support only consists in $\Fq$-rational points. We define the \emph{degree} of $D$ equals $\deg D \mydef \sum n_P$.

The set of divisors on the curve $\calX$ forms an additive group, denoted by $\Div(\calX)$. It is endowed with a partial order relation $\leq$ such that $D \leq D'$ if $D'-D$ is effective. An element $f$ of the function field $F\mydef \F(\calX)$ of the curve $\calX$ defines a divisor 
\[\div_{F}(f) = \sum_{P\in \calX} v_P(f) P\]
where $v_P(f)$ is the valuation of the function $f$ at the point $P$. The index $F$ will be omitted when the context is clear.

We denote by $\zeroes{f}$ (respectively $\poles{f}$) the positive (respectively negative) part of the principal divisor $\div(f)$, \textit{i.e.}
\[ \zeroes{f}=\sum_{\substack{P \in \calX\\ v_P(f) >0 }} v_P(f) P \quad \text{  and  } \quad \poles{f}=\sum_{\substack{P \in \calX\\ v_P(f) <0 }} v_P(f) P\]
so that $\div(f)=\zeroes{f} - \poles{f}$. The divisors $\zeroes{f}$ and $\poles{f}$ correspond to the loci of zeroes and poles respectively.

Let $\phi: \calX \rightarrow \calX'$ be a map between two algebraic curves. It induces a \emph{pull-back} map ${\phi^*:\F(\calX') \rightarrow \F(\calX)}$ defined by $\phi^* f =f \circ \phi$ for $f \in \F(\calX')$. For $D=\sum_P n_p P \in \Div(\calX)$, the \emph{push-forward} of $D$ is the divisor on $\calX'$ defined by $\pi_*(D)=\sum_P n_P \phi(P)$.

The \emph{Riemann-Roch space} of a divisor $D \in \Div(\calX)$ is the $\F$-vector space defined by
\[L_\calX(D)=\{f \in \F(\calX) \mid \div_{F}(f) + D \geq 0\}\cup\{0\}.\]
The subscript specifying the curve in $L_\calX(D)$ is omitted when it is clear from the context. If $D' \leq D'$, then $L_\calX(D) \subseteq L_\calX(D')$.

As usual, given a real number $\alpha$, $\floor{\alpha}$ denotes the biggest integer less than or equal to $\alpha$ and $\ceil{\alpha}$ the smallest integer bigger than or equal to $\alpha$.
\begin{definition}\label{def:floor_div}
	Let $D = \sum n_P P \in \Div(\calX)$. For any positive integer $n$, we denote by ${\floor{\frac{1}{n}D}\in \Div(\calX)}$ the divisor defined by 
	\[\floor{\frac{1}{n}D}:=\sum \floor{\frac{n_P}{n}} P.\]
\end{definition}

\subsection{Algebraic geometry codes}
Throughout this paper, the term $\emph{code}$ will refer to a $\emph{linear code}$, i.e. a linear subspace of $\F^n$, where $n$ is the length of the code.

Take $D \in \Div(\calX)$ and $\calP \subset \calX(\F)$ of size $n\mydef\size{\calP}$ such that $\Supp(D) \cap \calP = \emptyset$. The \emph{Algebraic Geometry (AG) code} $C=C(\calX,\calP,D)$ is defined as the image under the evaluation map
\[\text{ev} : L(D) \rightarrow \F^n.\]
The integer $n$ is called the \emph{length} of $C$. The \emph{dimension} of $C$ is defined as its dimension as $\F$-vector space. We denote by $\Delta(C)$ the relative minimum distance of $C$, i.e.
\[\Delta(C) = \min \set{\Delta(c, c') \mid c, c' \in C \text{ and } c \neq c'}.\]

In particular, AG codes on $\calX = \PP^1$ correspond to Reed-Solomon codes. 
The AG code $C$ is said to be \emph{one-point} if the support of $D$ consists in a single point.

By the {Riemann-Roch} theorem, if $\deg D \geq 2g-1$ where $g$ is the genus of the curve $\calX$, then $\dim L_\calX(D)=\deg D -g +1$. Moreover, if $\deg D< n$, the evaluation map is injective and the {Riemann-Roch} theorem gives the dimension of the associated AG code. In this case, the minimum distance is bounded from below by $n - \deg D$.

The divisor $D$ will always be chosen so that the map $\text{ev}$ is injective. Therefore, the elements of $\F^n$ will be regarded as functions in $\F^\calP$, and elements of $C$ simply as functions in the {Riemann-Roch} space $L(D)$.

\subsection{Group and action}
A finite group $\calG$  is said to be \emph{solvable} if  there exists a sequence of subgroups of $\calG$ 
\[\calG=\calG_0 \vartriangleright \calG_1 \vartriangleright \cdots \vartriangleright \calG_r = 1, \] 
such that $\calG_{i+1}$ is a normal subgroup of $\calG_i$ and each factor group $\calG_i / \calG_{i+1}$ is abelian. Such a sequence is called a \emph{normal series}. If $\calG$ is solvable, its cardinality equals the product of the sizes of the factor groups. 
Additionally, a normal series of $\calG$ is a \emph{composition series} of $\calG$ if each $\calG_{i+1}$ is a maximal proper normal subgroup of $\calG_{i}$. In the case where $\calG$ is a finite group, we have that $\calG$ is solvable if and only if $\calG$ has a composition series whose factor groups are cyclic groups of prime order.

Let $\calX$ be an algebraic curve. A group $\Gamma$ is said to \emph{act on the curve $\calX$} if $\Gamma$ is a subgroup of the automorphism group $\Aut(\calX)$. The \emph{stabilizer} of a point $P \in \calX$ is the subgroup \[\Gamma_P=\{\gamma \in \Gamma \mid \gamma\cdot P = P\} \subset \Gamma.\] 
A divisor $D=\sum_P n_P P \in \Div(\calX)$ is said to be \emph{$\Gamma$-invariant} is $n_P=n_{\gamma\cdot P}$ for all $P\in \calX$ and $\gamma \in \Gamma$.

The action of $\Gamma$ on $\calX$ gives a projection $\pi: \calX \rightarrow \calX/\Gamma$ onto the quotient curve $\calX/\Gamma$. A point $Q \in \calX/\Gamma$ is called a \emph{ramification point} if the number of preimages of $Q$ by $\pi$ is not equal to $\size{\Gamma}$. Equivalently, $Q$ is a ramification point if one of its preimages has a non-trivial stabilizer. 

\section{Setting of AG codes compatible with proximity test}\label{sec:foldable}

In this section, we display a workable setting for the construction of an IOPP system $(\prover, \verifier)$ to test whether a given function $f : \calP \rightarrow \F$ is close to the evaluation of a function in a given Riemann-Roch space. As the idea is to iteratively reduce the problem of testing proximity to $C(\calX, \calP, D)$ to testing proximity to a smaller AG code, we introduce a sequence of suitable AG codes of decreasing length.

\subsection{Sequence of curves}\label{subsec:seq-curves}
Fix a curve $\calX$ defined over $\F$, a finite solvable group $\calG \subseteq \Aut(\calX)$ and a sequence \[\seqG \mydef (\calG_0, \calG_1, \dots, \calG_r)\] such that 
\begin{equation}\label{diag:seq_groups}
	\calG=\calG_0 \vartriangleright \calG_1 \vartriangleright \cdots \vartriangleright \calG_r = 1,
\end{equation}
is a normal series for the group $\calG$.

For $i \in \range{r}$, we denote by $\Gamma_i$ the (abelian) factor group $\Gamma_i \mydef \calG_i /\calG_{i+1}$ and by $p_i$ the order of $\Gamma_i$. We have that the cardinality of $\calG$ equals \[\size{\calG}=\prod_{i=0}^{r-1} \size {\Gamma_i} =\prod_{i=0}^{r-1} p_i.\]

The group $\Gamma_0$ acts on $\calX_0\mydef \calX$, as a factor group of $\calG$. We thus define the quotient curve $\calX_1 \mydef \calX_0 / \Gamma_0$. The group $\Gamma_1$ acts trivially on the orbits under $\Gamma_0$. Repeating the process for every ${i \in \range{r}}$ defines a sequence of curves recursively  as follows:
\[\calX_0:=\calX \text{ and } \calX_{i+1} \mydef\calX_i/\Gamma_i.\]
We set $F_i \mydef \F(\calX_i)$ and we denote by $\pi_i: \calX_i \rightarrow \calX_{i+1}$ the canonical projection modulo the action of $\Gamma_i$.
Even if the sequence of curves \eqref{diag:seq_curves} depends on the derived series \eqref{diag:seq_groups} of $\calG$, the last curve $\calX_r$ is always isomorphic to the quotient $\calX / \calG$: \begin{equation}\label{diag:seq_curves}
	\begin{tikzcd}
		\calX_0\arrow[loop, distance=14, out=130, in=100,"\Gamma_0"] \arrow[r, two heads, "\pi_0"] & \calX_1\arrow[loop, distance=14, out=130, in=100,"\Gamma_1"] \arrow[r, two heads, "\pi_1"] & \cdots  \arrow[r, two heads, "\pi_i"] & \calX_i\arrow[loop, distance=14, out=130, in=100,"\Gamma_i"] \arrow[r, two heads, "\pi_{i+1}"]& \calX_{i+1}\arrow[loop, distance=14, out=130, in=100,"\Gamma_{i+1}"] \arrow[r, two heads] & \cdots \arrow[r, two heads,"\pi_{r-1}"] & \calX_{r}.
	\end{tikzcd}
\end{equation}

\begin{definition}
	A sequence of curves constructed as above will be called a \emph{$\XGseq$-sequence}.
\end{definition}

\subsection{Sequence of codes}\label{subsubsec:def-codes}

Let $(\calX_i)$ be a $\XGseq$-sequence. For any $i \in \range{r}$, the factor group $\Gamma_i$ which acts on the curve $\calX_i$ is abelian of order $p_i$.

For $i \in \set{0, \dots, r}$, we aim to define an AG code $C_i \subset \F^{\calP_i}$ associated to a divisor $D_i \in \Div(\calX_i)$ and an evaluation set $\calP_i$. 
The rest of this subsection is dedicated to the choice of the sets $\calP_i$ and the divisors $D_i$.

\subsubsection{Evaluation points} From a set $\calP_0 \subset \calX(\F)$, we want to recursively define a sequence of set of points $(\calP_i)$ such that $P_i \subseteq \calX_i(\F)$ and $\calP_{i+1}=\pi_i(\calP_{i})$.

For our protocol, we need for each $i \in \range{r}$ that every point in $\calP_{i+1}$ admits exactly $p_i$ preimages under $\pi_i$. Since the last curve $\calX_r$ is isomorphic to the quotient $\calX / \calG$, it is necessary and sufficient that the first set $\calP_0 \subset \calX_0$ is a union of $\calG$-orbits of size $\size{\calG}$, \emph{i.e.} that $\calG$ acts freely on $\calP_0$.

\subsubsection{Divisors} Fix a divisor $D_0 \in \Div(\calX_0)$ that is globally $\Gamma_0$-invariant. This way, the support of $D_0$ does not meet the set $\calP_0$. For the sake of simplicity, we will assume that $D_0$ is in fact supported by $\Gamma_0$-fixed points. To make our protocol complete and sound, we need the sequence of divisors $(D_i)$ to have the following properties:
\begin{itemize}
    \item the divisor $D_i$ is supported by $\Gamma_i$-invariant points;
    \item for each divisor $D_i$, its associated Riemann-Roch space admits a nice decomposition (given in \eqref{eq:decomposition} below);
    \item at each step, a divisor $D_{i+1}$ needs to be compatible with $D_i$ and the decomposition of $L_{\calX_i}(D_i)$ in the sense of Definition \ref{def-div_compatible} below.
\end{itemize}

\begin{definition}\label{def:mu-decomp}
	Let $i \in \range{r}$. Fix a divisor $D_i \in \Div(\calX_i)$ and a function $\mu_i \in F_i$ . We say that $\mu_i$ \emph{partitions} $L_{\calX_i}(D_i)$ (with respect to the action of $\Gamma_i$) if
	\begin{equation}\label{eq:decomposition}
		L_{\calX_i}(D_i)= \bigoplus_{j=0}^{p_i-1} \mu_i^j \pi_i^* L_{\calX_{i+1}}(E_{i,j})
	\end{equation}
	with
	\begin{equation}\label{eq:def-Eij}
		E_{i,j}\mydef \floor{\frac{1}{p_i}{\pi_i}_*( D_i + j\div_{F_i}(\mu_i))} \in \Div(\calX_{i+1}) \text{ for } j \in \range{p_i},
	\end{equation}
	where the floor function of a divisor is given in Definition \ref{def:floor_div}.
\end{definition}

\begin{definition}\label{def-div_compatible}
	Let $i \in \range{r}$. Fix a divisor $D_i \in \Div(\calX_i)$ and a function $\mu_i \in F_i$ such that $\mu_i$ partitions $L_{\calX_i}(D_i)$. A divisor $D_{i+1} \in \Div(\calX_{i+1})$ is said to be \emph{compatible with $(D_i,\mu_i)$} if both assertions hold.
	\begin{enumerate}
		\item for every $j \in \range{p_i}$, $E_{i,j} \leq D_{i+1}$,
		\item  for every $j \in \range{p_i}$, there exists a function $\nu_{i+1,j}\in \F(\calX_{i+1})$ such that
		\begin{equation}\label{eq:def-nu_i,j}
			\poles{\nu_{i+1,j}}=D_{i+1} - E_{i,j}. 
		\end{equation}
	\end{enumerate}
	The functions $\nu_{i+1,j}$ are called \emph{balancing functions}.
\end{definition}

In \autoref{def-div_compatible}, the first requirement implies that $L(E_{i,j}) \subseteq L(D_{i+1})$. The second one means that a function $f_j$ belongs to $L(E_{i,j})$ if and only if the product $\nu_{i+1, j}f_j$ lies in $L(D_{i+1})$.

We have now described all the key components to formally define the notion of foldable codes.

\begin{definition}[Foldable AG codes]\label{def:good_properties}
    Let $C =C(\calX,\calP,D)$ be an AG-code. This code is said to be \emph{foldable} if the following conditions are satisfied.
    \begin{enumerate}
        \item\label{it:gpe} There exists a finite solvable group $\calG \in \Aut(\calX)$ that acts freely on $\calP$ : a normal series of $\calG$ $\eqref{diag:seq_groups}$ provides a $\XGseq$-sequence of curves $(\calX_i)$;
        \item\label{it:card} There exists $e \in (0,1)$ such that $\size{\calG} > \size{\calP}^e$;
        \item\label{it:cdt} There exist some sequences of functions $(\mu_i)$ and divisors $(D_i)$ where $\mu_i \in F_i, D_0=D$ and  $D_i \in \Div(\calX_i)$ such that, for every $i \in \range{r}$, all the following properties hold:
        \begin{enumerate}
            \item the divisors $D_i$ are supported by $\Gamma_i$-fixed points,
            \item the function $\mu_i$ partitions $L_{\calX_i}(D_i)$ (Definition \ref{def:mu-decomp}),
            \item\label{sit:dif-values} the function $\mu_i$ maps distinct points in the same $\Gamma_i$-orbit to distinct values: for every $\gamma_i \in \Gamma_i$ and $P \in \calP_i$, we have $\mu_i(P)= \mu_i(\gamma_i(P))$ if and only if $\gamma_i$ is the identity map,
            \item  $D_{i+1}$ is $(D_i,\mu_i)$-compatible (Definition \ref{def-div_compatible}).			
        \end{enumerate}
        
    \end{enumerate}
\end{definition}

The second requirement given in Definition \ref{def-div_compatible} is definitely compelling and requires some geometric knowledge about the curves $\calX_i$. Indeed, on a general curve, not every effective divisor is the poles locus of a function. Characterizing which effective divisors arise this way is at the heart of the Weierstrass gaps theory. Nonetheless, the existence of the balancing functions $\nu_{i+1,j}$ happens to be the main ingredient in Lemma \ref{lem:why_nuij}, which will take a prominent role in the design of our IOPP.

To prevent the relative minimum distance of the code $C_{i+1}$ from collapsing and thence ensure a good soundness of the protocol designed in Section \ref{sec:iopp}, one may be tempted to take $D_{i+1}$ as one of the divisors $E_{i,j}$ \eqref{eq:def-Eij} that appear in the decomposition \eqref{eq:decomposition} of $L_{\calX_i}(D)$. However the Weierstrass gaps theory indicates that balancing functions exist only when choosing a $(D_i,\mu_i)$-compatible divisor $D_{i+1}$ whose degree may be unexpectedly substantial (an illustration of this situation will be given in Example \ref{ex:bad_kummer}). Therefore, to broaden the spectrum of foldable codes, we do not make this additional hypothesis.

Condition \ref{sit:dif-values} in \autoref{def:good_properties} is necessary for the completeness of our protocol. If the divisors $D_i$ have degree sufficiently large, it is always fulfilled, as proven in the following lemma.

\begin{lemma}\label{lem:mu-ok-interpolation}
	Let us assume that a function $\mu_i \in F_i$ partitions $L_{\calX_i}(D_i)$ with respect to the action of $\Gamma_i$ for some divisor $D_i \in \Div(\calX_i)$. Let us denote by $g_i$ the genus of the curve $\calX_i$. If $\deg(D_i)  \geq 2g_i+1$, then $\mu_i$ takes different values at points lying in the same orbit under $\Gamma_i$.
\end{lemma}

\begin{proof}
	Let us assume by contradiction that there exists two distinct points $P$ and $Q$ on $\calX_i$ such that $\pi_i(P)=\pi_i(Q)$ and $\mu_i(P)=\mu_i(Q)$. By Definition \ref{def:mu-decomp}, every function $f$ in $L(D_i)$ can be written $f=\sum\mu_i^j f_j \circ \pi_i$ some some functions $f_j$ on the quotient curve $\calX_{i+1}$. Then, for every $f \in L(D_i)$, we have $f(P)=f(Q)$. In particular, this means that a function of $L(D_i)$ vanishes at $P$ if and only if its vanishes at $Q$. Therefore, the Riemann-Roch spaces $L(D_i-P)$ and $L(D_i-P-Q)$ are equal, which contradicts the Riemann-Roch theorem, since $\deg(D_i-P-Q) \geq 2g_i-1$.
\end{proof}

\subsection{RS codes are foldable AG codes}\label{ex:RS-foldable}

Reed-Solomon codes are AG codes on the projective line $\PP^1$. Moreover the Riemann-Roch space $L_{\PP^1}(d P_\infty)$ on $\PP^1$ for $P_\infty=[0:1]$ is isomorphic to 
the set of polynomials of degree (less than or equal to) $d$.

In this case, the decomposition \eqref{eq:decomposition} is nothing but the splitting of a polynomial into an even part and an odd part, which plays a crucial role in the \textsf{FRI} protocol when the characteristic is not $2$.

To make both points of view coincide, let us consider the involution $\gamma :[X_0:X_1] \mapsto [-X_0 : X_1]$. It generates a group isomorphic to $\Z/2\Z$ and the quotient of $\PP^1$ by this group is obtained as the image by $\pi:[X_0:X_1]\mapsto [X_0^2:X_1^2]$.

The divisor $D\mydef d P_\infty$ is invariant under $\gamma$. Let us choose $\mu$ as the function $x= \frac{X_0}{X_1}$. We have $\div(x) = P_0 - P_\infty$ with $P_0=[1:0]$. Noticing that $\pi_*(P_\infty)=P_\infty$ and $\pi_*(P_0)=P_0$, we get
\[\floor{\frac{1}{2}\pi_*(D + (x))}=\floor{\frac{1}{2}((d-1)P_\infty + P_0)}=\floor{\frac{d-1}{2}}P_\infty,\]
and the Riemman-Roch space $L_{\PP^1}(d P_\infty)$ is split into two parts:
\[L_{\PP^1}(d P_\infty)= \pi^*L_{\PP^1}\left(\floor{\frac{d}{2}}P_\infty\right) + x \pi^*L_{\PP^1}\left(\floor{\frac{d-1}{2}}P_\infty\right).  \]
We recover the decomposition of a polynomial of degree $d$ into even and odd parts of respective degrees $\floor{\frac{d}{2}}$ and $\floor{\frac{d-1}{2}}$.

\begin{remark}
	The function $\mu$ is not unique: any odd polynomial of $x$ would make a suitable choice for $\mu$.
\end{remark}

\medskip

Now, let us remark that the RS code \[V \mydef \set{f \in \F^\calP ; \deg f \le d} = C(\PP^1, \calP, dP_\infty)\] is a foldable AG code, for any $\calP \subset \F$ of size $\size{\calP} = 2^{r}$ for a certain integer $r$ and any degree bound $d$. We shall then retrieve the construction of the RS proximity test of \cite{BBHR18a}. 

Firstly, the finite solvable group $\Z / 2^r \Z$ of size $\size{\calP}$ acts on $\PP^1$ via $[X_0:X_1] \mapsto [X_0,\xi X_1]$, where $\xi$ is a primitive $2^r$-th root unity. It clearly fulfils the two first items of Definition \ref{def:good_properties}. When considering its composition series

\begin{equation}
	\Z / 2^{r} \Z \vartriangleright \Z / 2^{r-1} \Z  \vartriangleright \cdots \vartriangleright 1
\end{equation}
and the action of the corresponding factor group $\Gamma = \gen{\gamma} \simeq \Z / 2 \Z$, we obtain a trivial sequence of curves $(\calX_i)$ with $\calX_i = \PP^1$ for all $i$. Next, consider the sequence $(\mu_i)$ with $\mu_i = \mu = x \mydef \frac{X_1}{X_0}$, then $\gamma \mu = - \mu$. Set $d_0 \mydef d$, and for any $i \in \range{r}$, $d_{i+1} \mydef \floor{\frac{d_i}{2}}$. Note that there exists $r' < r$ such that $d_{r'}, \dots, d_{r}$ are all equal to 0. 
Setting $D_i = \floor{\frac{d_i}{2}}P_\infty$, we have $\Gamma_i$-invariant divisors fulfilling the compatibility condition given in Definition \ref{def-div_compatible}, by letting $\nu_{i+1, j}$ to be the constant function equal to 1 if $\floor{\frac{d_i}{2}} = \floor{\frac{d_i-1}{2}}$, and $\nu_{i+1, j} : x \mapsto x$ otherwise.

\subsection{Splitting Riemann-Roch spaces according to a cyclic group of automorphisms}

The first requirement to make a sequence of codes foldable is the splitting the Riemann-Roch spaces as in \eqref{eq:decomposition}, which mimics the decomposition in odd and even parts of univariate polynomials.  Under some additional hypotheses, a decomposition like \eqref{eq:decomposition} always exists. Let us detail this framework.

\medskip

Let $\calX$ be a smooth irreducible curve over a field $\F$ and let $\Gamma$ be a \textit{cyclic} group of order $m$ generated by an element $\gamma$ that acts on $\F(\calX) \times \overline{\F}$. Assume that $m$ and the characteristic of $\F$ are coprime and consider $\zeta \in \bar{\F}$ a primitive $m^{th}$ root of unity, lying in some algebraic closure $\bar{\F}$ of $\F$. 

Set $\calY := \calX/\Gamma$ and $\pi : \calX \rightarrow \calY$ be the canonical projection morphism.

Fix a $\Gamma$-invariant divisor $D \in \Div(\calX)$. We want to exhibit a relation between the Riemann-Roch space $L_\calX(D)$ and some Riemann-Roch spaces on $\calY$. The group $\Gamma$ acts on the vector space $L_\calX(D)$ via $\gamma \cdot f = f \circ \gamma$. By representation theory,
\[L_\calX(D) = \bigoplus_{j=0}^{m-1} L_\calX(D)_j,\]
where $L_\calX(D)_j \mydef \{ g \in L_\calX(D) \mid \gamma \cdot g = \zeta^j g \}$.

One of the key ingredients of this section is a theorem due to Kani \cite{Kani}, which we reformulate here in the case where $\Gamma$ is cyclic.

\begin{theorem}[\cite{Kani}]\label{th:kani}
	Let $\Gamma=\langle \gamma \rangle$ be a cyclic group that acts on $\F(\calX) \times \overline{\F}$. Assume that the order $m = \size{\Gamma}$ is coprime with $\size{\F}$. Then the two following statements hold.
	\begin{enumerate}
		\item There exists a function $\mu \in \bar{\F}(\calX)$ such that $\gamma \cdot \mu = \zeta \mu$.
		\item For any $\Gamma$-invariant divisor $D \in \Div(\calX)$, when considering the Riemman-Roch spaces over the algebraic closure $\bar{\F}$, we have
		\begin{equation} \label{eq:kani-closure}
			{L_\calX(D)_j\otimes \bar{\F} \simeq \mu^j \pi^* \left(L_\calY\left(\floor{\frac{1}{m}\pi_*\left(D+j\div(\mu)\right)}\right)\otimes \bar{\F}\right)}.
		\end{equation}
	\end{enumerate}
	
\end{theorem}

\begin{remark}
	If the function $\mu$ is defined over the base field $\F$ then the decomposition \eqref{eq:kani-closure} is valid when considering $\F$-vector spaces:
	\[   {L_\calX(D)_j \simeq \mu^j \pi^* \left(L_\calY\left(\floor{\frac{1}{m}\pi_*\left(D+j\div(\mu)\right)}\right)\right)}.\] 
	In practical instantiations, we are always able to choose $\mu$ defined over $\F$, even when $\zeta$ does not belong to $\F$. For instance, in the case of Kummer curves (see Section \ref{subsec:Kummer}), the decomposition provided by Theorem \ref{th:M04} is always valid. However, as the evaluation set $\calP$ needs to be formed of orbits of size $\size \calG$, instantiations require this primitive root to belong to $\F$.\\
	Let us highlight that Theorem \ref{th:kani} does not apply on the Hermitian tower, as the degree of the Artin--Schrieir extensions is not coprime with the characteristic.
\end{remark}

To handle the divisors that appear in the decomposition above, we need to get a better grasp on the zeroes and the poles of the function $\mu$. If the group action fixes the field of constants, the ramification points of $\pi$ are zeroes or poles of the function $\mu$, as stated in the following lemma.

\begin{lemma}\label{lem:pts-ram}
	Assume that $\Gamma=\gen{\gamma}$ is a cyclic group of order $m$ which fixes the $m^{th}$ primitive root $\zeta$. Let $P$ be a point of $X$ whose stabilizer $\Gamma_P$ is non-trivial. Then ${P \in \Supp(\mu)}$.
\end{lemma}

\begin{proof}
	By hypothesis, there exists $j \in \Range{1}{m-1}$ such that $\gamma^j\in\Gamma_P$. Then
	\begin{align*}
		(\gamma^j\cdot \mu)(P)&=\zeta^j \mu(P) && \text{ by definition of } \mu\text{ in Th. \ref{th:kani},}&\\
		&= \mu(P) & &\text{ because }\gamma^j \in \Gamma_P.&
	\end{align*}
	Since $\zeta^j \neq 1$, the point $P$ is either a pole or a zero of $\mu$.
\end{proof}

\section{Foldable AG codes on Kummer curves}\label{subsec:Kummer}

\subsection{Preliminaries}

Let us consider a Kummer curve over a finite field $\F$ defined by an equation of the form 
\begin{equation}\label{eq:def-Kummer}\calX: y^N=f(x)=\prod_{\ell=1}^m (x-\alpha_\ell)\end{equation}
where $f$ is a degree-$m$ separable polynomial of $\F[X]$, $\gcd(N,m)=1$ and $\alpha_\ell \in \F$ for all $\ell \in \Range{1}{m}$. Let us denote by $P_\ell$ the affine point $(\alpha_\ell,0)$ and $P_\infty$ the unique point of $\calX$ lying on the line at infinity.

\subparagraph*{Sequence of curves.} Assume that $\gcd(N,\size{\F})=1$. 
The group $\Z/N\Z$ acts on $\calX$ via the morphism $(x,y) \mapsto (x,\zeta y)$ where $\zeta$ is a primitive $N^{th}$ root of unity. We assume that $\zeta$ belongs to $\F$.

The cyclic group $\Z/N\Z$ is solvable: writing the prime decomposition of $N=\prod_{i=0}^{s-1} p_i$ gives the following sequence of subgroups
\begin{equation}\label{eq:comp-series-kummer}
	\Z/N\Z \vartriangleright \Z/N_1\Z \vartriangleright \Z/N_2\Z \vartriangleright \dots \vartriangleright \Z/N_{s-1}\Z \vartriangleright 1, 
\end{equation}
where 
\begin{equation}\label{eq:def:N_i}
	N_i\mydef\prod_{j=i}^{s-1} p_j.
\end{equation}
The $i$-th factor group $\Gamma_i$ is isomorphic to the cyclic group of prime order $\Z/p_i \Z$. It is spanned by $\gamma_i:(x,y) \mapsto (x,\zeta_iy)$ where $\zeta_i$ is a primitive $p_i^{th}$ root of unity.

Set $\calX_0\mydef \calX$. By Section \ref{subsec:seq-curves}, the composition series \eqref{eq:comp-series-kummer} gives a sequence of curves $(\calX_i)$ in which the $i^{th}$ curve is defined by
\begin{equation}\label{Crvi}
	\calX_i : y^{N_i}=f(x)
\end{equation}
and has genus
\[g_i=\frac{(N_i-1)(m-1)}{2}.\]
The last curve $\calX_{s}$ has genus 0 and is isomorphic to the projective line $\PP^1$. These successive quotients provide a sequence of projections $\pi_i: \calX_i \rightarrow \calX_{i+1}$ defined by ${\pi_i(x,y)=(x,y^{p_i})}$:

\begin{center}
	\begin{tikzcd}
		{\calX_0 \arrow[loop, distance=14, out=130, in=100,"\gamma_0"]} \arrow[r, two heads, "\pi_0"] & \dots  \arrow[r, two heads, "\pi_i"] & \calX_i\arrow[loop, distance=14, out=130, in=100,"\gamma_i"] \arrow[r, two heads, "\pi_{i+1}"]& \calX_{i+1}\arrow[loop, distance=14, out=130, in=100,"\gamma_{i+1}"] \arrow[r, two heads] & \dots \arrow[r, two heads,"\pi_{r-1}"] & \calX_{r} \simeq \PP^1.
	\end{tikzcd}
\end{center} 

\begin{example}\label{ex:hermitian} 
	The Hermitian curve defined over $\F_{q^2}$ by
	\begin{equation}\label{eq:hermitian}
		\calX_0 : y^{q+1} = x^q+x.
	\end{equation}
	is a well-studied particular case of Kummer type curve. In this case, every curve in a $\XGseq$-sequence is maximal over $\F_{q^2}$ \cite[Proposition 6]{L87}, i.e. $\size{\calX_i(\F_{q^2})}=q^2+1+2g_iq$. 
\end{example}

\subparagraph*{Stabilized points.}

Let us denote $P^i_\infty$ the unique point at infinity on the curve $\calX_i$. One can easily check that 
\[P_\infty^i:=\left\{\begin{array}{cl}
	(1:0:0) &\text{if } N > m \\
	(0:1:0) & \text{otherwise.}
\end{array}\right.\]

The points of $\calX_0$ whose stabilizer under $\Z/N\Z$ is non-trivial are in fact fixed by $\Z/N\Z$ and consist precisely in $P_1,\dots,P_m$ and $P_\infty^i$.

\subsection{Decomposition of Riemann-Roch spaces}

The theory of Kummer extensions provides us a decomposition like \eqref{eq:decomposition} at each level, with $\mu_i=y$ for every $i\in\range{s}$.

\begin{theorem}[{\cite[Theorem 2.2]{M04}}]\label{th:M04}
	Let $D \in \Div(\calX_i)$ that is $\Gamma_i$-invariant. Then 
	\[L_{\calX_i}(D_i)=\bigoplus_{j=0}^{p_i-1}y^j L_{\calX_{i+1}}\left(\floor{\frac{1}{p_i}(\pi_i)_*(D_i+j\div_{F_i}(y))}\right).\]
\end{theorem}

Note that the function $y$ maps distinct points in the same $\Gamma_i$-orbits onto different values, and thus satisfies the condition \ref{sit:dif-values} of Definition \ref{def:good_properties}.

\subparagraph*{An example of a sequence of $y$-compatible divisors.}

In order to exhibit a sequence of divisors $(D_i)$ such that $D_{i+1}$ is $(D_i,y)$-compatible for every $i \geq 0$, we need to handle the divisor associated to $y$ and some other elementary functions on each curve $\calX_i$, described for instance in \cite{MQS15}.

\begin{lemma}[\cite{MQS15}]
	On $\calX_i$ for every $i \in \range{s}$, we have
	\begin{enumerate}
		\item $\div_{F_i}(x-\alpha_\ell)=N_i(P_\ell-P^i_\infty)$,
		\item $\div_{F_i}(y)=P_1+\dots+P_m - m P_\infty^i$.
	\end{enumerate}  
\end{lemma}

We now give sufficient conditions on the curve $\calX_0$ and the first divisor $D_0$ to get a sequence of compatible divisors.

\begin{lemma}\label{lem:compatible-Kummer} Set $\displaystyle{D_0= \sum_{\ell=1}^m a_{0,\ell} P_\ell + b_0 P^0_\infty \in \Div(\calX_0)}$.
	
	Assume that $m \equiv -1 \mod N$ and that the integers $a_{0,1},\dots,a_{0,m},b_{0}$ are all divisible by $N$. For every $i \in \range{s}$, set $D_{i+1}=\frac{D_i}{p_i}$. Then, the divisor $D_{i+1}$ is $(D_i,y)$-compatible.
\end{lemma}

\begin{proof}
	For $i \in \Range{1}{s}$, let us set $a_{i,\ell}=\frac{a_{i-1,\ell}}{p_{i-1}}$ and $b_i=\frac{b_{i-1}}{p_{i-1}}$ such that  $D_i=\sum_{\ell=1}^m a_{i,\ell} P_\ell + b_i P^i_\infty$.
	
	Fix $i \in \range{s}$. The divisor $D_i$ is supported only by $\Gamma_i$-fixed points.
	
	For any ${j \in \range{p_i}}$, we have
	\[E_{i,j} = \floor{\frac{1}{p_i} {\pi_i}_* (D_i+j\div_{F_i}(y))}=\sum_{\ell=1}^m \floor{\frac{a_{i,\ell} + j}{p_i}} P_\ell + \floor{\frac{b_i-jm}{p_i}}P^{i+1}_\infty.\]
	Since $N_i$ divides $N$, we have $m \equiv -1 \mod N_i$. Write $m=\kappa_i N_i -1$ with $\kappa_i \geq 1$. 
	The hypothesis on the integers  $a_{0,1},\dots,a_{0,m},b_{0}$ entails
	\[\begin{aligned}
		\floor{\frac{a_{i,\ell} + j}{p_i}}&=a_{i+1,\ell} + \floor{\frac{j}{p_i}}= a_{i+1,\ell} , \\
		\floor{\frac{b_i-jm}{p_i}}&=b_{i+1}-\frac{j\kappa_i N_i}{p_i} + \floor{\frac{j}{p_i}} = b_{i+1}-j\kappa_i N_{i+1} .
	\end{aligned}\]
	Then $E_{i,j}= D_{i+1} - j \kappa_i N_{i+1}P^{i+1}_\infty$. In particular, $D_{i+1}=E_{i,0}$ and $E_{i,j} \leq D_{i+1}$. Any ${\nu_{i+1,j}\mydef(x-\alpha)^{\kappa_i j}}$ with ${\alpha \in \{\alpha_1,\dots,\alpha_m\}}$ gives the last condition on $D_{i+1}$ for it to be $(D_i,y)$-compatible by Definition \ref{def-div_compatible}, i.e. ${D_{i+1}-E_{i,j}= \poles{\nu_{i+1,j}}}$.
\end{proof}

\subsection{Family of foldable codes}

We have gathered all the components to exhibit a foldable code on a family of Kummer curves.

\begin{proposition}\label{prop:Kummer_foldable}
	Let $\calX_0$ be a Kummer curve defined by \eqref{eq:def-Kummer} with $m \equiv -1 \mod N$. Consider an evaluation set $\calP_0\subseteq \calX_0(\F) \setminus \{P_1,\dots, P_m,P^0_\infty\}$ formed by $\Z/N\Z$-orbits.
	Take $D_0\in \Div(\calX_0)$ satisfying hypothesis of Lemma \ref{lem:compatible-Kummer}. If $N>n^e$ for some $e \in (0,1)$, then the AG code $C=C(\calX_0,\calP_0,D_0)$ is foldable.
\end{proposition}

The length of foldable codes over a Kummer curve as defined in \eqref{eq:def-Kummer} over $\F_q$ is bounded from above by $q+1+(N-1)(\kappa N-2)\sqrt{q}-\kappa N$, using Hasse-Weil bound, write $m=\kappa N-1$.

\begin{remark}
	\begin{enumerate}
		\item The primitive $N^{th}$ root $\zeta$ needs to belong to the base field $\F$ to ensure that the set $\calP_0$ is not empty.
		\item The condition on the coefficients of $D_0$ can be loosened while the previous statement still holds. If  $a_{0,1},\dots,a_{0,m},b_{0}$ are divisible by $\prod_{i=0}^{s-2} p_i$ and not necessarily by $p_{s-1}$, we choose $a_{s,\ell}=\ceil{\frac{a_{s-1,\ell}}{p_{s-1}}}\ \text{ and } b_{s}=\floor{\frac{b_{s-1}}{p_{s-1}}}$ for the coefficients of $D_s$. The last curve $\calX_s$ being isomorphic to $\PP^1$, the existence of balancing functions is trivial, if the first requirement of Definition \ref{def-div_compatible} holds.
	\end{enumerate} 
\end{remark}

\subparagraph*{What happens outside these hypotheses?}
Lemma \ref{lem:compatible-Kummer} provides sufficient conditions to make the code $C_{i+1}$ as small as possible compared to $C_i$ by choosing $D_{i+1}$ among the divisors $E_{i,j}$, as required for a sequence of foldable codes by Definition \ref{def:good_properties}. Let us have a look at what could happen when dropping these conditions.

\begin{example}\label{ex:bad_kummer}
	Over $\F_8$, consider $y^N=x^m+x$ where $N=9$ and $m=5$. Then $m \not\equiv -1 \mod N$ and $N=p_0p_1$ with $p_0=p_1=3$. For $D_0=18 P_\infty^0$, we have
	\begin{align*}
	E_{0,0}&=\floor{\frac{18}{3}}P_\infty^1=6 P_\infty^1, \\
	E_{0,1}&=\floor{\frac{18-5}{3}} P_\infty^1=4 P_\infty^1, \\ 
	E_{0,2}&=\floor{\frac{18-2 \times 5}{3}} P_\infty^1=2 P_\infty^1.
    \end{align*}
	
	Choosing $D_1=E_{0,0}$ would satisfy the first and the second conditions of Definition \ref{def-div_compatible} to be $(D_0,y)$-compatible but not the third one. One can reasonably ask the support of $D_1$ to consist only of $\pi_0(P_\infty^0)=P_\infty^1$, as one-point codes are generally better understood. The Weierstrass gap theory on Kummer curves (e.g. \cite[Theorem~3.2]{MQS15}) entails that if a function on $\calX_1 : y^3=x^5+x$ has a pole locus of the form $\alpha P_\infty^1$, then $\alpha \in 3\Z_++5\Z_+$. Therefore the smallest divisor of the form $D_1=d_1P_\infty^1$ that is $(D_0,y)$-compatible is $D_1=12 P_\infty^1$. With such a choice of divisors, the code $C_0$ of dimension 15 is folded into the code $C_1$ of dimension 12 whereas the length of $C_1$ is the third of the length of $C_0$.
	
\end{example}

\subsubsection{Explicit basis of the Riemann-Roch spaces}

AG codes from the Kummer curve $\calX$ associated to divisors as defined in Lemma \ref{lem:compatible-Kummer} have been studied by Hu and Yang \cite{HY18}. They provide a basis of the Riemann-Roch spaces in a combinatorial form.

\begin{theorem}[{\cite[Theorem~5]{HY18}}]
	Let $j, \: j_2, \dots, j_m$ be integers. We define
	\[E_{j, \: j_2, \dots, j_m}\mydef y^j \prod_{\ell=2}^m (x-\alpha_\ell)^{j_\ell}.\]
	Consider $D=\sum_{\ell =1}^m a_\ell P_\ell + bP_\infty$. Set
	\begin{align*}\Omega_{a_1,\dots,a_m,b}\mydef \biggl\{ (j, \: j_2, \dots, j_m) \mid j+a_1 \geq 0, \: j_\ell =\ceil{\frac{-j-a_\ell}{N}} \text{ for } \ell = 2,\dots,m   \\
		\text{ and }mj+N(j_2+\dots+j_m) \leq b \biggr\}.\end{align*}
	
	Then the elements $E_{j, \: j_2, \dots, j_m}$ for $(j, \: j_2, \dots, j_m) \in \Omega_{a_1,\dots,a_m,b}$ form a basis of $L_\calX(D)$.
\end{theorem}

\subsubsection{Parameters}

To estimate the parameters of the code by using the Riemann-Roch theorem, we shall rely on the following result. 

\begin{lemma}\label{lem:nice_deg}
	Assume that $2(g_0-1) < \deg(D_0)$  (resp. $\deg(D_0)< n_0$). Then for every $i \in \Range{0}{s}$, $2(g_i-1) < \deg(D_i)$ (resp. $\deg(D_i) < n_i$).
	
\end{lemma}

\begin{proof}
	It is enough to notice that for every $i \in \range{s}$, 
	\[ \begin{aligned}
		&\deg D_{i+1} = \frac{\deg{D_i}}{p_i},\qquad  & n_{i+1} = \frac{n_i}{p_i}, \qquad & \text{ and} & g_{i+1} \leq \frac{g_i}{p_i}.&
	\end{aligned}\]
\end{proof}

In other words, if the degree of the first divisor is such that we can estimate the parameters of $C_0$ thanks to Riemann-Roch Theorem, then we handle the parameters of all the sequence of codes.

\begin{proposition}\label{prop:dist_min_Kummer}
	If $\deg(D_0) < n_0$, then for every $i \in \Range{0}{s}$, the code $C_i$ has length $n_i$ and minimum relative distance $\Delta(C_i)=1- \frac{\deg D_0}{n_0}$. In particular, the RS code $C_s$ has length $\frac{n_0}{N}$, dimension $\frac{\deg D_0}{N} + 1$ and relative minimum distance $1- \frac{\deg D_0}{n_0}$.
	
	Moreover, if $2(g_0-1) < \deg(D_0)$, for every $i \in \Range{0}{s}$, the code $C_i$ has dimension $\deg D_i -g_i +1$.
\end{proposition}

\begin{proof}
	The length of $C_i$ is $n_i$ by construction and its dimension is given by the Riemann-Roch theorem. So let us prove the statement concerning the relative minimum distance.
	
	First notice that $n_i=p_i n_{i+1}$  and $\deg(D_i)=p_i \deg(D_{i+1})$ so $1- \frac{\deg D_i}{n_i}=1- \frac{\deg D_0}{n_0}$. 
	For $i=s$, the code $C_s$ is a Reed-Solomon code of degree $0 \leq \deg(D_s) < n_s$ by Lemma \ref{lem:nice_deg} and has the expected relative minimum distance.
	
	Now assume that $\Delta(C_{i+1})$ equals $1- \frac{\deg D_{0}}{n_{0}}$ and let us prove that so does $\Delta(C_i)$. 
	On the one hand, the divisor $D_{i+1}$ corresponds to $E_{i,0}$ then for every $f \in C_{i+1}$, $f \circ \pi_i \in C_i$. In addition, the weight of $f \circ \pi_i$ in $C_i$ is $p_i$ times the weight of $f$ in $C_{i+1}$. Since $n_i=p_i n_{i+1}$, we have $\Delta(C_i) \leq \Delta(C_{i+1})$. On the other hand, as $\deg(C_i) < n_i$, we have $\Delta(C_i) \geq 1- \frac{\deg{D_i}}{n_i}$, which concludes the proof. 
\end{proof}

\section{Foldable AG codes along the Hermitian tower}\label{subsec:HermitianTower}

\subsection{Preliminaries}

\subparagraph*{Sequence of curves.}
We consider the sequence of function fields $\calF=(F_i)_{i\geq 0}$ over $\F_{q^2}$ that is defined recursively by $F_0=\F_{q^2}(x_0)$ and $F_i = F_{i-1}(x_i)$ with equations
\begin{equation}\label{eq:herm_tower}
	x_i^q+x_{i} = x_{i-1}^{q+1} \text{ for } i\geq 1.
\end{equation}
Note that this tower of function field $\mathcal{F}$ corresponds to a tower of curves
$(\calX_i)_{i\geq 0}$ such that $F_i=\F_{q^2}(\calX_i)$. One can view the curve $\calX_i$ embedded in an $i$-dimensional affine space with variables $(x_1,\dots,x_i)$ defined by the equations \eqref{eq:herm_tower}.

For $i=1$, the field $F_1$ is the function field of the Hermitian curve $\mathcal{H}:= \calX_1$ over $\F_{q^2}$.

Let $g_i:=g(\calX_i)$ denote the genus of the curve $\calX_i$. An explicit formula was given by Pellikaan, Shen and Wee \cite[Proposition 4]{PSW91}: we have $g_0=0$ and for $i\geq 1$, 

\begin{equation}\label{eq:genus_herm}
	g_i = \frac{1}{2}\left[(q^2-1)\left((q+1)^i-q^i\right) +1 -q^i\right]
	=\dfrac{1}{2} \cd \left(\sum\limits_{k=1}^i q^{i+1} \cd \left(1+\frac{1}{q}\right)^{k-1} +1 -(1+q)^i\right).
\end{equation}

For every $i \geq 0$, the number of $\F_{q^2}$--rational places in $F_i$ is given by
\[\size{\calX_i(\F_{q^2})} = q^{i+2}+1.\]

We have an infinite sequence of curves $(\calX_i)_{i\ge 0}$ as follows.

\begin{center}
	\begin{tikzcd}
		\dots  \arrow[r, two heads, "\pi_{i+1}"] & \mathcal{X}_i \arrow[loop, distance=14, out=130, in=100,"\Gamma_i"] \arrow[r, two heads, "\pi_{i}"]& \mathcal{X}_{i-1}\arrow[loop, distance=14, out=130, in=100,"\Gamma_{i-1}"] \arrow[r, two heads, "\pi_{i-1}"] & \dots \arrow[r, two heads,"\pi_{1}"] & \mathcal{X}_{0} \simeq \PP^1.
	\end{tikzcd}
\end{center} 

\begin{remark}
	In the context of recursive towers, it is classical to \textbf{index the curves the other way} round compared to the notations used in Section \ref{sec:foldable}. In the context of the Hermitian tower, the curve $\calX_{i-1}$ is the quotient curve of the curve $\calX_i$ under the action of the group $\Gamma_i$. To design foldable codes, we will fix a level of the tower, say $\ii$, and consider codes over the curve $\calX_{\ii}$, which we will fold using the sequence of curves with \emph{decreasing} indices $(\calX_i)_{\ii \geq i \geq 0}$.
\end{remark}

This tower is a tower of Artin-Schreier extensions, which have been extensively studied (see for example \cite{S08}).  Let us recall some classical results that will be useful to design foldable AG codes along this tower.

\subparagraph*{Automorphisms and projection maps.}
By definition of the Hermitian tower \cite[Proposition 3.7.10]{S08}, the Galois group of the extension $F_i/F_{i-1}$ is the group of automorphisms defined by $(x_1, \dots, x_{i-1}, x_i) \mapsto (x_1, \dots, x_{i-1}, x_i + \alpha)$ where $\alpha$ runs in 
\[\mathcal{S}  = \set{\alpha \in \F_{q^2} \mid \alpha^q + \alpha = 0}.\]
Note that if we fixed a non--zero element $\alpha \in \mathcal{S}$, then for every $\beta \in \Fq$, $\alpha \beta$ lies in $\mathcal{S}$. So $\mathcal{S}$ is an additive group which is isomorphic to $\Fq$.

The quotient map $\pi_i : \calX_i \rightarrow \calX_{i-1}$ consists in the projection onto the first $i$ coordinates. For every $i\geq 0$, we set $\Pi_i$ to be the composition of the first $i$ quotient maps, \textit{i.e.}
\[ \Pi_i \mydef \pi_i \circ \pi_{i-1}\circ \dots \circ \pi_0.\]

\subparagraph*{Behaviour of the point of infinity.}

In what follows, let us denote by $P_{\infty}^{(0)}$ the unique pole of the function $x_0$ in $F_0$, which corresponds to the point at infinity on the projective line $\calX_0=\PP^1$.

\begin{lemma}[{\cite[Proposition 3.7.8]{S08}}]
	Let $i \geq 1$. The place $P_{\infty}^{(0)}$ is totally ramified in $F_i$, which means that the preimage $\Pi_i^{-1}\left(\left\{P_{\infty}^{(0)}\right\}\right)$ consists in a unique place, denoted by $P_{\infty}^{(i)} \in \calX_i$.
	Moreover,  $P_{\infty}^{(0)}$ is the unique place that is ramified in the tower $\calF$.
\end{lemma}

The peculiar behaviour of the points $P_\infty^{(i)}$ in the tower encourages us to define a sequence of codes associated with divisors $D_i \in \Div(\calX_i)$ of the form 
\[D_i\mydef d_i P_\infty^{(i)} \text{ for }i \geq 1.\]

Let us focus on the principal divisors $\div_{F_i}(x_j)$ ($0\leq j \leq i$) and their valuation at the point $P_\infty^{(i)}$. 
Their properties follow from the study of the basic function field $F=\mathbb{F}_{q^2}(x,y)$ which is nothing but the Hermitian function field. It is a special case of Artin-Schreier extension of $\mathbb{F}_{q^2}(x)$ and is well known that we have 
\[\div_{F}(y) = P^{(0)}-P_{\infty}^0.\]
\begin{remark}
	The role of the variables $x$ and $y$ is reversed compared to the Kummer model of the Hermitian curve, studied in the previous section. 
\end{remark}
Since each extension $F_i/F_{i-1}$ corresponds to the same Artin-Schreier extension, and that $P_{\infty}^{(0)}$ is fully ramified in $F_i/F_0$, we can deduce the form of the divisor $\div_{F_i}(x_i)$, given in the next lemma. The valuation of the function $x_j \in F_{i-1}$ at $P_{\infty}^{(j)}$ follows from the extension degrees $[F_{i-1}:F_j]=q^{i-1-j}$ for $j<i$.

\begin{lemma} \label{lem:div-princ} The following two assertions hold.
	\begin{enumerate}
		\item For $i \geq 1$, denote by $P^{(i)}$ the unique common zero of the functions $x_0,...,x_i$. Then we have
		\[\div_{F_i}(x_i) = (q+1)^i \left(P^{(i)} - P_{\infty}^{(i)}\right);\]
		\item Let $i \geq 1$. Then for $0 \leq j < i$, the valuation of the function $x_j \in F_{i-1}$ is given by
		\[v_{P_{\infty}^{(i-1)}}(x_j) = - q^{i-1-j}(q+1)^j.\]
	\end{enumerate}
\end{lemma}

\subparagraph*{Basis of the Riemann-Roch spaces associated to the divisor $d_i P_\infty^{(i)}$.}

For a given $i \geq 0$, the $P^{(i)}_{\infty}$ is the unique pole of the functions $x_0,...,x_i$, which gives an explicit basis of the Riemann-Roch space associated to a multiple of $P^{(i)}_{\infty}$.

\begin{lemma}\label{lem:basis_herm}
	For all $i \leq 1$ and $m \leq 1$, the Riemann-Roch space $L_{\calX_i}(mP^{(i)}_{\infty})$ is formed by linear combinations of functions in the following set:
	\[\left\{x_0^{a_0} \cdots x_i^{a_i} \ | \ 0 \leq a_0 \ , \ 0 \leq a_j \leq q-1 \text{ and } \sum\limits_{j=0}^i a_jq^{i-j}(q+1)^j \leq m \right\}.\]
\end{lemma}

\subsection{Construction of foldable AG codes}

Fix a level $\ii$ in the Hermitian tower. We want to construct foldable codes on the curve $\calX_{\ii}$ of the form 
\[C(\calX_{\ii},\calP_{\ii},D_{\ii}) \text{ where } \calP_{\ii} \subseteq \calX_{\ii}(\F_{q^2}) \setminus \{P_\infty^{(\ii)}\} \text{ and } D_{\ii} = d_{\ii}P^{(\ii)}_{\infty}.\]

We thus define a sequence of codes $(C_i)$ as follow:
\[C_i \mydef C(\calX_{\ii-i},\calP_{\ii-i},D_{\ii-i}) \text{ where } \calP_{i-1}=\pi_i(\calP_i)\text{ and } D_i = d_iP^{(i)}_{\infty} \]

In order to make sure $C_0=C(\calX_{\ii},\calP_{\ii},D_{\ii})$ is foldable, we need to describe the Riemann-Roch spaces on a certain step from Riemann-Roch spaces on \emph{lower} curves. A priori Kani's theorem does not apply, so we will have to find a decomposition by hand. We deduce such a decomposition  from the explicit basis of the Riemann-Roch space given in Lemma \ref{lem:basis_herm}.

\begin{proposition} 
	Let $i \geq 0$. Set $D_i = d_iP^{(i)}_{\infty}$ for some integer $d_i$. Then
	\[L_{\calX_i}(D_i) = \bigoplus\limits_{j=0}^{q-1} x_i^j \pi_i^*\left(L_{\calX_{i-1}}(E_{i,j})\right)\]
	with
	\[E_{i,j} \mydef \left\lfloor \frac{1}{q} \pi_{i*}\left(D_i-j\cdot \div_{F_i}(x_i)\right)\right\rfloor \text{ for } 0 \leq j \leq q-1.\]
	In other words, the function $x_i \in F_i$ partitions the divisor $D_i$ in the sense of Definition \ref{def:mu-decomp}.
\end{proposition}

\begin{proof}
	By Lemma \ref{lem:basis_herm}, the space $L_{\calX_i}(D_i)$ is formed by linear combinations of $x_0^{a_0} \cdots x_i^{a_i}$ with non-negative exponents  such that $0 \leq a_j \leq q-1$ for $j\neq 1$ and \[\sum\limits_{j=0}^i a_jq^{i-j}(q+1)^j \leq m.\] As $a_j$ runs in $\{0,\dots,q-1\}$, the proof is concluded by noticing that the function $x_0^{a_0} \cdots x_{i-1}^{a_{i-1}} \in F_{i}$ lies in $L(D_i-j\cdot\div_{F_i}(x_i))$ which means that $x_0^{a_0} \cdots x_{i-1}^{a_{i-1}} \in F_{i-1}$ belongs to $L_{\calX_{i-1}}(E_{i,j})$.
	
\end{proof}

To make $D_{i-1}$ compatible with $(D_i,x_i)$ (Definition \ref{def-div_compatible}), we need the existence of $q$ balancing functions  $\nu_{i-1,j} \in F_{i-1}$ (for every  $0 \leq j \leq q-1$) such that 
\begin{equation}\label{eq:balancing_fct_herm}
	D_{i-1}-E_{i,j} = (\nu_{i-1,j})_{\infty}.
\end{equation}
In our setup, we have
\[E_{i,j} = \left\lfloor \dfrac{d_i-j(q+1)^i}{q}\right\rfloor P^{(i-1)}_{\infty}.\]
Thus, we need to ``balance'' the divisors
\[D_{i-1}-E_{i,j} = \left(d_{i-1}- \left\lfloor \dfrac{d_i-j(q+1)^i}{q}\right\rfloor\right) P^{(i-1)}_{\infty}.\]

We are led to study the Weierstrass semigroup of $P^{(i-1)}_{\infty}$, denoted by $\calH\left(P^{(i-1)}_{\infty}\right)$. The generators of this semigroup can be found using Lemma \ref{lem:div-princ}. In fact, $P_{\infty}^{(i-1)}$ is the unique common pole of the functions $x_0,...,x_{i-1} \in 
F_{i-1}$ and we know their exact valuation. Thus we have

\[\calH\left(P^{(i-1)}_{\infty}\right) = \left \langle q^{i-1-k}(q+1)^k \ , \ 0\leq k \leq i-1 \right\rangle_{\N}.\]

\begin{remark}\label{rk:big-div-HT}
	In the spirit of the \textsf{FRI} protocol, one could be tempted to choose $D_{i-1}$ as $E_{i,0}$. Such a choice would be valid in the sense of Definition \ref{def-div_compatible} if and only for every $0 \leq j \leq q-1$
	\[ \left\lfloor \dfrac{d_i}{q} \right\rfloor - \left\lfloor \dfrac{d_i-j(q+1)^i}{q}\right\rfloor \in \calH\left(P^{(i-1)}_{\infty}\right).\]
	Unfortunately, when $i$ increases, this condition is never satisfied.
\end{remark}

To ensure that 
$\deg(D_{i-1}-E_{i,j})$ is never a Weierstrass gap for $P_{\infty}^{(i-1)}$, we increase the degree $d_{i-1}$ of $D_{i-1}$. 

\begin{theorem}\label{th:div-comp-HT}
	Let $i \geq 1$. Set $D_i = d_i P^{(i)}_{\infty}$ for some integer $d_i$ and $D_{i-1} = d_{i-1} P^{(i-1)}_{\infty}$ where
	\begin{equation}\label{eq:deg_herm}
		d_{i-1} \mydef \floor{ \dfrac{d_i}{q}} + 2g_{i-1}.
	\end{equation}
	Then $D_{i-1}$ is compatible with $(D_i,x_i)$  (Definition \ref{def-div_compatible}).
\end{theorem}

\begin{proof}
	By \cite[Theorem 1.6.8]{S08}, we know that
	\[\max\left(\N \setminus \calH\left(P^{(i-1)}_{\infty}\right)\right) \leq 2g_{i-1} -1.\]
	Then for every $0 \leq j \leq q-1$, the difference
	\begin{equation}
		m_{i,j}:=\deg(D_{i-1}-E_{i,j}) =  \left(\left\lfloor \dfrac{d_i}{q} \right\rfloor - \left\lfloor \dfrac{d_i-j(q+1)^i}{q}\right\rfloor + 2g_{i-1} \right)
	\end{equation}
	always belongs to the Weierstrass semigroup $\mathcal{H}\left(P_{\infty}^{(i-1)}\right)$.
\end{proof}

\subparagraph*{About the balancing functions.}

Since we know a $\mathbb{N}$-basis of the Weierstrass semigroup at $P_{\infty}^{(i-1)}$, we are able to explicit the form of the functions $\nu_{i-1,j}$. In particular, they can be chosen as the products of powers of the functions $x_0,...,x_{i-1}$. More precisely, if $a_{i,j} := (a_{i,j}(0),...,a_{i,j}(i-1)) \in \mathbb{N}^i$ are integers such that

\begin{equation}\label{eq:mij}
	m_{i,j} = \sum\limits_{k=0}^{i-1} a_{i,j}(k) \cd q^{i-1-k}(q+1)^k,
\end{equation}
then $m_{i,j} \in \mathcal{H}\left(P_{\infty}^{(i-1)}\right)$. The corresponding choice for the balancing function is then given by 
\[ \nu_{i-1,j} = \prod\limits_{k=0}^{i-1} x_k^{a_{i,j}(k)}.\]
Note that finding a vector $a_{i,j} \in \mathbb{N}^i$ satisfying \eqref{eq:mij} leads to the study of the diophantine equation
\[m_{i,j} = \sum\limits_{k=0}^{i-1} a_k \cd q^{i-1-k}(q+1)^k\]
with $i$ unknowns $a_k \in \mathbb{N}$, for which we know there exists at least a solution (and we only need one).

\subsubsection{A family of foldable codes}\label{subsec:foldable-towers-codes}

We denote by $\ii$ the level in the tower $(\calX_i)_{i\geq 0}$, such that $\calX_{\ii}$ is the curve on which the code we want to test proximity is defined.

\begin{proposition}\label{prop:tower-family}
	Fix an integer $\ii$. Set $\calP_{0} \subseteq \PP^1(\F_{q^2})\setminus\{P^{0}_\infty\}$ and define $\calP_{\ii}$ as the preimage of $\calP_0$ under the morphism $\Pi_{\ii}$. Fix an integer $d_{\ii}$. Then the code $C(\calX_{\ii},\calP_{\ii}, d_{\ii}P^{(i)}_{\infty})$ is foldable.

\end{proposition}
\begin{proof}
	Beware that the sequence of curves is indexed decreasingly here, contrary to Section \ref{sec:foldable}. In particular, the original code is defined over the curve $\calX_{\ii}$.
	
	By definition of the Hermitian tower, there exists a solvable group $\calG$ acting on $\calX_{\ii}$ that admits a normal series for which each factor group is isomorphic to the additive abelian group of $\Fq$. The action of $\calG$ on $\calP_{\ii}$ is free, by definition of $\calP_{\ii}$. The cardinality of $\calP_{\ii}$ is equal to $\size{\calP_{0}}q^{\ii}$, hence $\size{\calG} > \size{\calP_{\ii}}^e$ for some $e \in (0,1)$. 
	
	The third condition of Definition \ref{def:good_properties} follows from Theorem \ref{th:div-comp-HT}, noticing that for every $i\geq 0$, the function $x_i$ maps different points in the same $\Gamma_i$-orbit onto different values.
\end{proof}

To control the dimension of foldable codes, we will focus on those of the form
\begin{equation}\label{eq:def_Cimax}
	C \mydef C\left(\calX_{\ii},\calX_{\ii}(\F_{q^2})\setminus \{P^{(\ii)}_{\infty} , (2\alpha+1)g_{\ii})P^{(\ii)}_{\infty}\right)
\end{equation}
for some $\alpha > 1/2$. In this case, we have $n_{\ii}=q^{\ii+2}$. We  can determine a sufficient condition over $\ii$ and $\alpha$ to get a constant rate.

\begin{lemma}\label{lem:limit-rate}
	Let $R \in (0,1)$. Fix $\epsilon \in (0,1)$. Set $\ii\mydef q^\epsilon$ and $\alpha \mydef R q^{1-\epsilon}$.
	
	The ratio of the dimension of the code $C$ defined in \eqref{eq:def_Cimax} by its block length goes to $R$ when $q$ tends to infinity. 
	
	If $2(q^\epsilon-1)<q$, the relative minimum distance of $C$ is bounded from below by $1-R\left(1+\frac{1}{q}\right)$.
\end{lemma}

\begin{proof}
	If $\alpha > \frac{1}{2}$, the dimension of the code is equal to $(2\alpha +1)g_{\ii}- g_{\ii}+1= 2Rq^{1-\epsilon}g_{\ii}+1$ by the Riemann-Roch Theorem. As $R$ is fixed and $q$ goes to infinity, we can assume that $\alpha > 1/2$ to compute the rate as
	\[\lim_{q \rightarrow \infty} \frac{2Rq^{1-\epsilon}g_{\ii}}{q^{\ii+2}}\]
	for $\ii=q^\epsilon$.
	Lemma \ref{lem:gen_sympt} in Appendix \ref{app:gen_HT} clearly implies that this limit is equal to $R$.
	
	Regarding the relative minimum distance, we use the Goppa bound: if \[(2\alpha + 1)g_{\ii} < q^{\ii+2},\] then the relative minimum distance of $C$ satisfies \[\Delta(C) \geq 1 - \frac{(2\alpha + 1)g_{\ii}}{q^{\ii+2}}.\] By Proposition \ref{prop:maj_genus_herm} in Appendix \ref{app:gen_HT}, we have
	\[ \frac{g_{\ii}}{q^{\ii+2}} \leq \frac{\ii}{2q}\left(1 + \frac{\ii}{q}\right),\]
	which gives the expected lower bound for our choice of $\alpha$ and $\ii$.
\end{proof}

\section{Folding operators for AG codes}\label{sec:iopp}

Now that we have determined the needed properties of an AG-code to be foldable, we construct the fundamental building block of our IOPP by generalizing the so-called algebraic hash function of \cite{BKS18} to the AG codes setting, and we refer to it as the \emph{folding operator}. Next, we provide a formal description of the IOPP system $(\prover, \verifier)$ and state the theorem capturing its efficiency properties.

\subsection{Definition of folding operators}\label{subsec:folding}

Let $C_0 = C(\calX_0,\calP_0,D_0)$ be a code satisfying Definition \ref{def:good_properties}. We consider its associated $\XGseq$-sequence of curves $(\calX_i)$ and its sequence of divisors $(D_i)$. 

To test proximity of a function ${f^{(0)} : \calP_0 \rightarrow \F}$ to $C_0$, we aim to inductively reduce the problem to a smaller one, consisting of testing proximity to the code $C_i = C(\calX_i,\calP_i,D_i)$. Broadly speaking, our goal is to define from any function $f^{(i)} : \calP_i \rightarrow \F$ a function $f^{(i+1)} : \calP_{i+1} \rightarrow \F$ such that the relative distance $\Delta(f^{(i+1)}, C_{i+1})$ is roughly equal to $\Delta(f^{(i)}, C_i)$.

Fix $i \in \range{r}$ and let $f : \calP_i \rightarrow \F$ be an arbitrary function. 

\begin{notation}[Interpolation polynomial]\label{not:IfP}
	For each $P \in \calP_{i+1}$, let us denote ${S_P \mydef \pi_i^{-1}(\set{P})}$ the set of $p_i$ distinct preimages of $P$.     Recall that the function $\mu_i$ satisfies Item \ref{sit:dif-values} of Definition \ref{def:good_properties} and consider 
	\begin{equation}\label{eq:IfP}
		I_{f,P}(X)\mydef\sum_{j=0}^{p_i-1} X^j a_{j,P} 
	\end{equation} 
	the univariate polynomial over $\F$ of degree less than $p_i$ which interpolates the set of points \[\set{(\mu_i(\hat P), f(\hat P)) ; \hat P \in S_P}.\] 
	Specifically, for all $\hat P \in S_P$, we have
	\[
	I_{f, P}(\mu_i(\hat P)) = f(\hat P).
	\]    
	Then for every $j \in \range{p_i}$, we define the function 
	\begin{equation}\label{eq:coeff-fj}
		f_j:\left\{\begin{array}{ccc}
			\calP_{i+1} &\rightarrow& \F, \\
			P &\mapsto & a_{j,P}.
		\end{array}\right.
	\end{equation}
\end{notation}

Given $f : \calP_i \rightarrow \F$, the idea is to define $p_i$ functions $f_j : \calP_{i+1} \rightarrow \F$, where $\size{\calP_{i+1}} = \frac{\size{\calP_i}}{p_{i}}$ such that $f$ corresponds to the evaluation of a function in $L(D_i)$ if and only if each $f_j$ coincides with a function in $L(E_{i, j}) \subset L(D_{i+1})$. Instead of testing for each $j \in \range{p_i}$ whether $f_j \in C_{i+1}$, we reduce those $p_i$ claims to a single one, by taking a random linear combination of the $f_j$'s, which we referred to as a \emph{folding of $f$}. By linearity of the codes, such a combination of the $f_j$'s belongs to $C_{i+1}$ whenever $f \in C_i$ (see Proposition \ref{prop:folding-completeness} below). However, for soundness analysis,  one needs to ensure that no $f_j$ corresponds to a function lying in $L(D_{i+1}) \setminus L(E_{i,j})$. Some safeguards are embedded into the folding operation by introducing the balancing functions $\nu_{i+1,j}$ from Definition \ref{def-div_compatible} in the second term of the sum in \eqref{eq:fold-via-interp}.

\begin{definition}[Folding operator]\label{def:folding-operator}
	For any $\bfz=(z_1,z_2) \in \F^2$, we define the \emph{folding of $f$} to be the function ${\fold{f, \bfz} : \mathcal P_{i+1} \rightarrow \F}$ such that
	\begin{equation}\label{eq:fold-via-interp}
		\fold{f, \bfz} \mydef\sum_{j=0}^{p_i-1} z_1^j f_j + \sum_{j=0}^{p_i-1} z_2^{j+1} \nu_{i+1,j} f_j
	\end{equation}    
	where the functions $f_j$ are defined in Equation \eqref{eq:coeff-fj} and the functions $\nu_{i+1,j}$ in Definition \ref{def-div_compatible}.
\end{definition}

\subsection{Properties of folding operators}
Our aim is to prove that the folding operators satisfy three key properties: local computability, completeness, and distance preservation. This will enable us to invoke \cite[Theorem~1]{ABN21} for the completeness and soundness of our AG-IOPP.

Given the $p_i$ points $((\mu_i(\hat P), f(\hat P)))_{\hat P \in S_P}$, one can determine the coefficients $(a_{j, P})_{0 \le j < p}$ of $I_{f, P}$ defined in \eqref{eq:IfP} by polynomial interpolation. Recalling that for each $P \in \calP_{i + 1}$, we have ${f_j(P) = a_{j, P}}$, we get the following lemma. This lemma will allow to obtain fast prover time and verifier decision complexity.

\begin{lemma}[Locality]\label{lem:locality}
	Let $\bfz \in \F^2$. For each $P \in \calP_{i+1}$, the value of $\fold{f, \bfz}(P)$ can be computed with exactly $p_i$ queries to $f$, namely at the points $\pi_i^{-1}(\set{P})$.
\end{lemma}

\begin{proposition}[Completeness]\label{prop:folding-completeness}
	Let $\bfz \in \F^2$. If $f \in C_i$, then $\fold{f,\bfz} \in C_{i+1}$.
\end{proposition}

\begin{proof}
	Write $\bfz=(z_1,z_2)$. If $f \in C_i$, it coincides with a function of $L(D_{i})$. By definition of the divisors $E_{i,j}$ and Theorem \ref{th:kani}, there exist some functions $\widetilde{f_j} \in L(E_{i,j})$ such that
	\[ f= \sum_{j=0}^{p_i-1} \mu_i^j \widetilde{f_j} \circ \pi_i .\]
	Let $P \in \mathcal{P}_{i+1}$.    
	For any $\hat P \in S_P$,
	\[\fold{f,(\mu_i(\hat P),0)}(P) = I_{f, P}(\mu_i(\hat P)) = f(\hat P)=\sum_{j=0}^{p_i-1} \mu_i(\hat P)^j \widetilde{f_j}(P).\]
	Moreover, for any $P \in \calP_{i+1}$, polynomials $I_{f, P}(X)$ and $\fold{f, (X,0)}(P)$ in $\F[X]$ are of degree less than $p_i$ and agree on $\set{\mu_i(\hat P) ; \hat P \in S_P}$ of size $p_i$, therefore they are equal. In particular, 
	\[
	\fold{f, (\mu_i(\hat P),0)}(P) = \sum_{j = 0}^{p_i - 1} \mu_i(\hat P)^j f_j(P).
	\]
	Thus, for all $P \in \calP_{i+1}$, \[\sum_{j=0}^{p_i-1} \mu_i(\hat P)^j (\widetilde{f_j}(P) - f_j(P)) = 0\] and the polynomial \[\sum_{j=0}^{p_i-1} X^j (\widetilde{f_j}(P) - f_j(P))\] of degree less than $p_i$ is zero on at least $\size{\set{\mu_i(\hat P) ; P \in \calP_{i+ 1}}} = p_i$ points. Hence, for every ${j \in \range{p_i}}$, the function $f_j$ defined in Equation \eqref{eq:coeff-fj} coincides with $\widetilde{f_j}$ and
	\[\fold{f, \bfz} \mydef\sum_{j=0}^{p_i-1} z_1^j \widetilde{f_j} + \sum_{j=0}^{p_i-1} z_2^{j+1} \nu_{i+1,j} \widetilde{f_j}\]
	where $\widetilde{f_j} \in L(E_{i,j}) \subseteq L(D_{i+1})$ and $\nu_{i+1,j} f_j \in L(D_{i+1})$, by definition of the divisors $E_{i,j}$, $D_{i+1}$ and the functions $\nu_{i+1,j}$ (see Definition \ref{def-div_compatible}). Thus each term of $\fold{f,\bfz}$ lies in the vector space $C_{i+1}$, which concludes the proof.
\end{proof}

We discuss the effect of the folding operation on a function which is far from the code. Roughly speaking, we want to show that, if $f$ is $\delta$-far from the code $C_i$, then the folding $\fold{f, \bfz}$ of $f$ is almost $\delta$-far from the code $C_{i+1}$ with high probability over $\bfz \in \F^2$. We start with the notion of weighted agreement.

\begin{definition}[Weighted agreement]\label{def:agreement}
	For any function $\eta \in [0,1]^{\calP}$, we define the \emph{$\eta$-agreement} of two functions $u,v \in \F^\calP$ by
	\[\omega_\eta(u,v) \mydef\frac{1}{|\calP|} \sum_{\mathclap{\stackrel{P \in \calP}{u(P) = v(P)}}} \eta(P).\]
	Given a subspace $V \subset \F^\calP$ and $u \in \F^\calP$, we set \[\omega_\eta(u,V) \mydef \max_{v \in V} \omega_\eta(u,v).\]
\end{definition}
Notice that since $\eta \in [0,1]^\calP$, we have for any $V \subset \F^\calP$ and any $u \in \F^\calP$,
\begin{equation}\label{eq:eta<1}
	\omega_\eta(u,V) \leq 1 - \Delta(u,V).
\end{equation}

Let us introduce some notations related to the Johnson list-decoding function. For any $\epsilon \in (0,1]$, let $J_\epsilon : [0,1] \rightarrow [0,1]$ be the function such that \[J_\epsilon(\lambda) = 1 - \sqrt{1 - (1 - \epsilon)\lambda},\] and denote $J_{\epsilon}^{l} = \underbrace{J_\epsilon \circ \dots \circ J_\epsilon}_{\text{$l$ times}}$.

We now state a preliminary result concerning the weighted agreement on a low-degree parametrized curve. Proof of Proposition \ref{prop:BKS18-eta} builds upon the one of \cite[Theorem~4.5]{BKS18} and is given in Appendix \ref{app:prop-ldc}.

\begin{proposition}\label{prop:BKS18-eta}
	Let $\eta \in [0,1]^\calP$ and $\epsilon,\delta>0$ such that and $\delta < J_{\epsilon}^l(\lambda)$.
	Let ${u_0, \ldots, u_{l-1} \in \F^\calP}$ such that
	\begin{equation}\label{eq:pr-randcombi-weighted}
		\Pr_{z \in \F }\left[\omega_\eta\left(\sum_{i = 0}^{l - 1} z^i u_i, V\right)>1-\delta \right]\geq \frac{l-1}{|\F|}\left(\frac{2}{\epsilon}\right)^{l+1},
	\end{equation}
	then there exists $T \subset \calP$ , and $v_0, \ldots, v_{l-1} \in V$ such that:
	\begin{itemize}
		\item $\sum_{P\in T} \eta(P) \ge (1 - \delta - \epsilon)|\calP|$
		\item for each $i$, $\rest{u_i}{T} = \rest{v_i}{T}$.
	\end{itemize}
\end{proposition}

Here, for a function $u \in \F^\calP$, $\rest{u}{T} \in \F^T$ corresponds to the function obtained by restriction on $T \subset \calP$.

As mentioned earlier, soundness analysis relies on the relation between the weighted agreement of $f$ to $C_i$ and the weighted agreement of the folding of $f$ to $C_{i+1}$, constrained by the next corollary.

\begin{corollary}\label{coro:pre-soundness}
	Fix $i \in \range{r}$. For a function $\eta : \calP_i \rightarrow [0,1]$, define ${\theta : \calP_{i+1} \rightarrow [0,1]}$ by
	\[\forall P \in \calP_{i+1}, \: \theta(P)\mydef\frac{1}{p_i} \sum_{\hat{P} \in S_P} \eta(\hat{P}).\] 
	Let $\lambda_i$ be the minimal relative distance of $C_i$. Fix $\epsilon \in (0,1)$ and \[\delta < \min\left(J_{\epsilon}^{p_i}(\lambda_i), \frac{1}{2}\left( \lambda_i + \frac{\epsilon}{2} \right) \right).\] For any function $f : \calP_i \rightarrow \F$ such that $\omega_\eta(f, C_i) < 1-\delta$, we have
	\[
	\Pr_{\bfz \in \F^2}\left[\omega_\theta(\fold{f, \bfz}, C_{i+1}) > 1 - \delta + \epsilon\right] \leq 
	\frac{1}{\size \F}\left(p_i + \frac{4}{\epsilon} - 1\right)\left(\frac{4}{\epsilon}\right)^{p_i}.
	\]
\end{corollary}

Proving Corollary \ref{coro:pre-soundness} requires the lemma stated next. We prove Corollary \ref{coro:pre-soundness}, then prove Lemma \ref{lem:why_nuij}.

\begin{lemma}\label{lem:why_nuij}
	Let $i \in \range{r}$, $D_i \in \Div(\calX_i)$ and $\mu_i \in \F(\calX_i)$ satisfying Definition \eqref{def:mu-decomp}. Consider a divisor $D_{i+1}\in\Div(\calX_{i+1})$ that is $(D_i,\mu_i)$-compatible in the sense of Definition \ref{def-div_compatible}.
	
	Fix $j \in \range{p_i}$. Then a function $g\in \F(\calX_{i+1})$ belongs to $L(E_{i,j})$ if and only if both functions $g$ and $g \nu_{i+1,j}$ belong to $L(D_{i+1})$.
\end{lemma}

\begin{proof}[Proof of Corollary \ref{coro:pre-soundness}]
	Let $f : \calP_i \rightarrow \F$ be an arbitrary function. According to Equation \eqref{eq:coeff-fj}, there exist $p_i$ function $f_j : \calP_{i+1} \rightarrow \F$ such that for any $\bfz=(z_1,z_2) \in \F^2$,
	\[\fold{f, \bfz} = \sum_{j = 0}^{p_i - 1}z_1^j f_j + \sum_{j = 0}^{p_i - 1}z_2^{j+1} \nu_{i+1,j} f_j.\]
	Rewrite $\fold{f, \bfz}$ as a polynomial function in $z_2$: \[\fold{f, \bfz}=f_{z_1} + z_2 f'_0+  z_2^2 f'_1 + \dots + z_2^{p_i} f'_{p_i-1},\] where \[f_{z_1} \mydef \sum_{j = 0}^{p_i - 1}z_1^j f_j, \qquad f'_j\mydef \nu_{i+1,j} f_j.\] Finally, set
	\[
	K_0 \mydef \frac{p_i - 1}{\size \F}\left(\frac{4}{\epsilon}\right)^{p_i} \qquad \text{ and } \qquad K_1 \mydef \frac{p_i}{\size \F}\left(\frac{4}{\epsilon}\right)^{p_i + 1} .
	\]    
	Let us prove the corollary by contrapositive. We assume that
	\[
	\Pr_{\bfz \in \F^2}\left[\omega_\theta(\fold{f, \bfz}, C_{i+1}) > 1 - \delta + \epsilon\right] > K_0 + K_1,
	\]
	and thus \[\Pr_{z_1 \in \F}\left[\Pr_{z_2 \in \F} \left[\omega_\theta(\fold{f, \bfz}, C_{i+1}) > 1 - \delta + \epsilon\right] > K_0\right] > K_1.\]

	Fix $z_1 \in \F$ such that \[\Pr_{z_2 \in \F} \left[\omega_\theta(\fold{f, \bfz}, C_{i+1}) > 1 - \delta + \epsilon\right] > K_0.\] By Proposition \ref{prop:BKS18-eta}, there exist $v_{z_1},v'_1,\dots,v'_{p_i-1} \in C_{i+1}$ and $\calT' \subset \calP$ such that
	\begin{itemize}
		\item $\sum_{P \in \calT'} \theta(P)  \ge (1 - \delta+ \frac{\epsilon}{2})\size{\calP_{i+1}}$,
		\item $\rest{v_{z_1}}{\calT'}=\rest{f_{z_1}}{\calT'}$,
		\item for each $j \in \Range{1}{p_i-1}$, $\rest{v'_j}{\calT'} = \rest{f'_j}{\calT'}$.
	\end{itemize}
	In particular, \[\omega_\theta(f_{z_1},C_{i+1}) \geq \omega_\theta(f_{z_1},v_{z_1}) =  \frac{1}{\size{\calP_{i+1}}}\sum_{P \in \calT'} \theta(P) \geq 1- \delta + \frac{\epsilon}{2}.\]   
	It means that
	\begin{align*}  \Pr_{z_1 \in \F} \left[\omega_\theta(f_{z_1},C_{i+1}) \geq 1-\delta+ \frac{\epsilon}{2} \right] &\geq \Pr_{z_1 \in \F}\left[ \Pr_{z_2 \in \F} \left[\omega_\theta(\fold{f, \bfz}, C_{i+1}) > 1 - \delta + \epsilon\right]  > K_0 \right] \\ &> K_1.
    \end{align*}
	The polynomial form of $f_{z_1}$ in $z_1$ enables us to reapply Proposition \ref{prop:BKS18-eta}: there exist $\calT \subset \calP$ and $v_0,v_1,\dots,v_{p_i-1} \in C_{i+1}$ such that
	\begin{itemize}
		\item $ \sum_{P \in \calT} \theta(P) \ge (1 - \delta)\size{\calP_{i+1}}$,
		\item for each $j \in \range{p_i}$, $\rest{v_j}{\calT} = \rest{f_j}{\calT}$.
	\end{itemize}
	On $\calT' \cap \calT$, we thus have \[\rest{v'_j}{\calT' \cap \calT}= \rest{f'_j}{\calT' \cap \calT}=\rest{(\nu_{i+1,j}f_j)}{\calT' \cap \calT}=\rest{(\nu_{i+1,j}v_j)}{\calT' \cap \calT}.\]
	The cardinality of $\calT' \cap \calT$ satisfies
	\[|\calT' \cap \calT| =|\calT'|+|\calT|-|\calT' \cup \calT| \geq \sum_{P \in \calT'} \theta(P)+\sum_{P \in \calT} \theta(P)-|\calP_{i+1}| \geq (1-2\delta+\frac{\epsilon}{2})|\calP_{i + 1}|.\]
	The assumption on $\delta$ ensures that $2\delta - \frac{\epsilon}{2} < \lambda_{i+1}$ where $\lambda_{i+1}$ is the minimal distance of $C_{i+1}$. Hence, for every $j \in \range{p_i}$, the evaluations of $v'_j$ and $\nu_{i+1,j}v_j$ on $\calP_{i+1}$ are equals. They are codewords of $C_{i+1}$, thus this implies that both functions $v_j$ and $\nu_{i+1,j}v_j$ belong to $L(D_{i+1})$. By Lemma \ref{lem:why_nuij}, we get that the function $v_j$ lies in $L(E_{i,j})$.
	
	Now let us define $v : \calP_i \rightarrow \F$ by 
	\[\forall Q \in \calP_i, \: v(Q) \mydef \sum_{j=0}^{p_i-1} \mu_i^j(Q) v_j \circ \pi_i(Q).\]
	By definition of the divisors $E_{i,j}$ \eqref{eq:def-Eij}, the function $v$ belong to $L(D_i)$. Now let us prove that it agrees with $f$ on $S_{\calT} \mydef \bigsqcup_{P \in \calT} S_P$.
	
	Let $P \in \calT$ and $\hat{P} \in S_P$. 
	\begin{align*}
		f(\hat P) &= I_{f, P}(\mu_i(\hat P))= \sum_{j=0}^{p_i-1} \mu_i(\hat{P})^j f_j (P) & &\text{ by definition of } I_{f,P},& \\
		&= \sum_{j=0}^{p_i-1} \mu_i(\hat{P})^j v_j \circ \pi_i(\hat{P}) & &\text{ since } \rest{f_j}{\calT} = \rest{v_j}{\calT} \text{ and } P=\pi_i(\hat{P}),& \\
		& = v(\hat P).&& &
	\end{align*}
	
	As a result, since $v \in C_i$, we can conclude that
	\[\omega_\eta(f,C_i) \geq \omega_\eta(f,v) \geq \frac{1}{\size{\calP_i}} \sum_{P \in \calT} \sum_{\hat{P} \in S_P} \eta(\hat{P}) = \frac{1}{\size{\calP_{i+1}}}   \sum_{P \in \calT} \theta(P) \geq 1-\delta. \]
\end{proof}

\begin{proof}[Proof of Lemma \ref{lem:why_nuij}] Assume that $g \in L(E_{i,j})$. Then Definition \ref{def-div_compatible} ensures that $g$ and $g\nu_{i+1,j}$ lie in $L(D_{i+1})$.
	
	Conversely, assume that $g$ and $g\nu_{i+1,j}$ belong to  $L(D_{i+1})$ and write $D_{i+1} = \sum n_P P$. The hypotheses on $g$ imply that \[g \in L(D_{i+1}) \cap L(D_{i+1} - (\nu_{i+1,j})).\] By \cite[Lemma~2.6]{MP93}, the function $g$ belongs to $L(D'_{i+1})$, where the divisor $D'_{i+1}$ is defined by
	\[D'_{i+1} \mydef \sum_P n'_P P \quad \text{ where } n'_P \mydef \min(n_P,n_P+v_P(\nu_{i+1,j})).\]
	Then $D'_{i+1} = D_{i+1} - \poles{\nu_{i+1,j}} = E_{i,j}$ by the second item of Definition \ref{def-div_compatible}.
\end{proof}

\subsection{IOPP for foldable AG codes}\label{subsec:iopp-overview}

Let $C_0 = C(\calX_0, \calP_0, D_0)$ be a foldable AG code over an alphabet $\F$. 
Given a family of folding operators defined as per Definition \ref{def:folding-operator}, \cite{ABN21} yields an IOPP for $C_0$, which is abstracted from the \textsf{FRI} protocol of \cite{BBHR18a}. We informally describe the IOPP system $(\prover, \verifier)$ for testing proximity of a function $f^{(0)} : \calP_0 \rightarrow \F$ to $C_0$, then give its properties. A formal description will be provided in Section \ref{sec:kummer-iopp} for instantiations with concrete AG codes.

As in the \textsf{FRI} protocol, the IOPP is divided in two phases, referred to as \textsf{COMMIT} and \textsf{QUERY}. 
Before any interaction, $\prover$ and $\verifier$ agree on:
\begin{itemize}
	\item a $\XGseq$-sequence of curves $(\calX_i)$, for which we denote the length of the composition serie of $\calG$ by $r$.
	\item a sequence of codes $(C_i)$ where for each $i \in \set{0, \dots, r}$, $C_i = (\calX_i, \calP_i, D_i)$ and $\calX_i, \calP_i$ and $D_i$ are defined as per Section \ref{sec:foldable},
	\item a sequence of functions $(\mu_i) \in \F(\calX_i)$ satisfying Definition \ref{def:mu-decomp},
	\item a sequence of balancing functions $(\nu_{i+1})_{0 \le i < r}$ of $p_i$-tuples of functions in $\F(\calX_{i+1})$ such that $\nu_{i+1} = (\nu_{i+1, j})_{0 < j < p_i}$ and $\nu_{i+1, j}$ satisfies \eqref{eq:def-nu_i,j}.
\end{itemize}
We recall that the choice of a sequence $(\calX_i)$ induces a sequence of projections $\pi_i : \calX_i \rightarrow \calX_{i+1}$.

\begin{itemize}
	\item The \textsf{COMMIT} phase 
	is an interaction over $r$ rounds between $\prover$ and $\verifier$. For each round $i \in \range{r}$, the verifier samples a random challenge $\bfz^{(i)} \in \F^2$. As an answer, the prover gives oracle access to function $f^{(i+1)} : \calP_{i+1} \rightarrow \F$, which is expected to be equal to $\fold{f^{(i)}, \bfz^{(i)}}$. To compute the values of $f^{(i+1)}$ on $\calP_{i+1}$, an honest prover $\prover$ exploits the fact that the folding of $f^{(i)}$ is locally computable (Lemma \ref{lem:locality}). Namely, for each $P \in \calP_{i+1}$, $\prover$ computes the coefficients $(a_{j,P})_{0 \le j < p}$ of  $I_{f^{(i)}, P} \in \F[X]$ from $\rest{f^{(i)}}{S_P}$, evaluates $\nu_{i+1,j}$ at $P$, and sets 
	\[
	\fold{f^{(i)}, \bfz^{(i)}}(P) \mydef\sum_{j=0}^{p_i-1} \left(z_1^{(i)}\right)^j a_{j, P} + \sum_{j=0}^{p_i-1} \left(z_2^{(i)}\right)^{j+1} \nu_{i+1,j}(P) a_{j, P}.
	\] 
	\item During the \textsf{QUERY} phase, one of the two tasks of the verifier $\verifier$ is to check that each pair of successive oracle functions $(f^{(i)}, f^{(i+1)})$ is consistent. A standard idea is to check that the equality 
	\begin{equation}\label{eq:test-equality}
		f^{(i+1)} = \fold{f^{(i)}, \bfz^{(i)}}
	\end{equation}
	holds at a random point in $\calP_{i+ 1}$.
	By leveraging the local property of the folding operator, such a test requires only $p_i$ queries to $f^{(i)}$ and 1 query to $f^{(i + 1)}$. As in \cite{BBHR18a}, we call this step of verification a \emph{round consistency test}. The verifier begins by sampling at random $Q_0 \in \calP_0$ and once this is done, all the locations of the round consistency tests run inside the current \textsf{\textbf{query test}} are determined. More specifically, for each round $i$,  $\verifier$ defines $Q_{i+1} \mydef \pi_i(Q_i)$ to be the random point where Equation (\ref{eq:test-equality}) is checked. Through this process, the round consistency tests are correlated to improve soundness. Such a \textsf{\textbf{query test}} can be seen as a \emph{global} consistency test, similar to the one of the \textsf{FRI} protocol. For the final test, $\verifier$ reads $f^{(r)} : \calP_r \rightarrow \F$ in its entirety to test if $f^{(r)} \in C_r$.
\end{itemize}

\begin{theorem}\label{thm:properties}
	Let ${C_0 = C(\calX_0, \calP_0, D_0)}$ be a foldable AG code of length $n \mydef \size{\calP_0}$. By definition, $C_0$ admits a solvable group $\calG \in \Aut(\calX_0)$ such that $\size{\calG} > n^e$ for a certain $e \in (0,1)$ and induces a sequence of codes $(C_i)$. Set $p_{max}\mydef \max p_i$, ${\lambda \mydef \min_i \Delta(C_i)}$ and \[\gamma \mydef \min\left(J_{\epsilon}^{p_{max}}(\lambda), \frac{1}{2}(\lambda + \frac{\epsilon}{2})\right).\]
	There is an IOPP system $(\prover, \verifier)$ for $C_0$ satisfying:
	\begin{description}
		\item \textbf{\emph{Perfect completeness:}} If $f^{(0)} \in C_0$ and $f^{(1)}, \dots, f^{(r)}$ are honestly generated by the prover, the verifier outputs $\mathsf{accept}$ with probability 1.
		\item \emph{\textbf{Soundness:}} Assume $f^{(0)}$ is $\delta$-far from $C_0$ and let $\epsilon \in (0,1)$. With probability  at least $1 - \err_{commit}$ over the randomness of the verifier during the \textsf{COMMIT} phase, where 
		\[
		\err_{commit} \leq \frac{\log n}{\size \F}\left(p_{max} + \frac{4}{\epsilon} - 1\right)\left(\frac{4}{\epsilon}\right)^{p_{max}}
		\] 
		and for any oracles $f^{(1)}, \ldots, f^{(r)}$ adaptively chosen by a possibly dishonest prover $\prover^*$, the probability that the verifier $\verifier$ outputs $\mathsf{accept}$ after a single $\textsf{query test}$ is at most
		\[
		\err_{query}(\delta) \leq (1 - \min(\delta, \gamma) + \epsilon \log n).
		\]        
		Overall, for any prover $\prover^*$, the soundness error
		$\err(\delta)$ after $t$ repetitions of the \textsf{QUERY} phase satisfies
		\begin{align*}
			\err(\delta) &\leq \err_{commit} + \left(\err_{query}(\delta)\right)^t\\
			&< \frac{\log n}{\size \F}\left(p_{max} + \frac{4}{\epsilon} - 1\right)\left(\frac{4}{\epsilon}\right)^{p_{max}} + (1 - \min(\delta, \gamma) + \epsilon\log n)^t.
		\end{align*}
	\end{description}
	
	Moreover, the IOPP system is public-coin, has round complexity $\mathsf{r}(n) < \log n$, proof length $\mathsf{l}(n) < n$ and query complexity ${\mathsf{q}(n) < t p_{max}\log n + n^{1 - e}}$.
\end{theorem}

\begin{proof} Lemma \ref{lem:locality}, Proposition \ref{prop:folding-completeness} and Corollary \ref{coro:pre-soundness} satisfy the conditions of \cite[Theorem~1]{ABN21}. Completeness and soundness are given by \cite[Theorem~1]{ABN21}. Let us prove the rest of the theorem.  
	Regarding round complexity, we have that \[\prod_{i = 0}^{r-1} p_i = \frac{n}{n_r},\] where \[n_r = \size{\calP_r} = \frac{n}{\size{\calG}} < n^{1 - e}.\] For every $i \in \range{r}$, $2 \le p_i \le p_{max}$. Therefore \[\mathsf{r}(n) \le \log_2 n - \log_2 n_r < \log_2 n.\]  
	
	For query complexity, notice that for $i \in \Range{0}{r-2}$, $\xpip{f}(Q_{i+1})$ is reused for the next round consistency test.  Hence, \[\mathsf{q}(n) = t \left(\sum_{i = 0}^{r-1} p_i\right) + n^{1 - e} \le t r p_{max} + n^{1 - e}.\]

	Finally, the total proof length $\mathsf{l}(n)$ is the sum of the lengths of all the oracles provided by $\prover$ during the \textsf{COMMIT} phase, counted in field elements. Denoting $t_{i+1} \mydef \prod_{j = 0}^{i} p_j$, we notice that \[\size{\calP_{i+1}} = \frac{\size{\calP_i}}{p_i} = \frac{\size{\calP_0}}{t_{i+1}}.\] Thus, we have
	\[\mathsf{l}(n) = \sum_{i = 1}^{r}\size{\calP_{i}} = \sum_{i = 1}^{r} \frac{\size{\calP_0}}{t_i} \le n \sum_{i = 1}^{r} \frac{1}{2^{i}} = n \left(1 - \frac{1}{2^r}\right) < n.\]      
\end{proof}

\section{Proximity tests for AG codes on Kummer curves and Hermitian towers}\label{sec:kummer-iopp}

When we instantiate the AG-IOPP proposed in Section \ref{subsec:iopp-overview} for the setting of Kummer curves (Section \ref{subsec:Kummer}) and curves in the Hermitian tower (Section \ref{subsec:HermitianTower}), we end up with a membership test to a RS code. An RS code is itself a foldable AG code (see Section \ref{ex:RS-foldable}). In order to lower verifier complexity, we can extend the AG-IOPP by replacing the final test by an IOPP for RS code. This enhanced AG-IOPP is examined in this section.

\subsection{How to iterate the folding to reach a code of dimension 1}\label{subsec:fast-iopp}
Let $C=C(\calX,\calP,\calG)$ be foldable in the sense of Definition \ref{def:good_properties} on a Kummer curve or on a curve in the Hermitian tower. In Sections \ref{subsec:Kummer} and \ref{subsec:HermitianTower}, we defined $s$ codes $(C_i)_{0\le i\le s}$, where $s$ is the the number of prime in the decomposition of $N$ in the Kummer case and $s=\ii$ for the Hermitian tower. The code $C_s = C(\PP^1, \calP', D')$ corresponds to a Reed-Solomon code $\RS{\calP', d} = \set{f : \calP_s \rightarrow \F; \deg f \le d}$, where the degree bound depends on the parameters of the original code $C_0$. Taking this into consideration, we want to iterate the folding operation until we get a RS code of dimension 1, as it is done in the \textsf{FRI} protocol \cite{BBHR18a}. 

As in Example \ref{ex:RS-foldable}, we set $d_0 = d$ and define $d_{i+1} = \floor{\frac{d_i}{2}}$ for any integer $i$. Set $s'$ the smallest integer such that $d_{s'} = 0$. Then, we consider the sequence of Reed-Solomon codes $(C_{s + i})_{1 \le i \le s'}$ when applying the construction described in Section \ref{sec:foldable} to the initial code $C_s$. Letting $r = s + s'$, we iteratively reduce the proximity test to the code $C_0$ to a membership test to the code $C_{r}$, which is a Reed-Solomon code of dimension 1. If $f^{(0)} \in C_0$, then $f^{(r)}$ is expected to be a constant function, and this can be tested in a trivial way. We can leverage the fact that $C_r$ is a Reed-Solomon code to extend the protocol described in Section \ref{subsec:iopp-overview}. We obtain a $r$-round IOPP system $(\prover, \verifier)$ for $C_0$, which is described below.

The prover $\prover$ and the verifier $\verifier$ are given as input the description of the code $C_0$.      
The verifier $\verifier$ is given oracle access to a function $f^{(0)} : \calP_0 \rightarrow \F$, which is also given as explicit input to the prover $\prover$.      

\paragraph*{\textsf{COMMIT} phase:}
\begin{enumerate}
	\item For each round $i$ from 0 to  $r-1$ :
	\begin{enumerate}
		\item $\verifier$ picks uniformly at random $\bfz^{(i)}$ in $\F^2$ and sends it to $\prover$,
		\item $\prover$ computes $\xpip{f} = \fold{\xpi{f}, \xpi{\bfz}}$,
		\item If $i < r - 1$: $\prover$ gives oracle access to $f^{(i + 1)} : \calP_{i + 1} \rightarrow \F$.
		\item If $i = r - 1$: $\prover$ commits to $\beta \in \F$ (if $f^{(0)} \in C_0$, then $f^{(r)}$ is supposed to be constant equal to $\beta$).                
	\end{enumerate}
\end{enumerate}

\paragraph*{\textsf{QUERY} phase:}
\begin{enumerate}
	\item Repeat $t$ times the following \textbf{\textsf{query test}}: 
	\begin{enumerate}
		\item Pick $Q_0 \in \calP_0$ uniformly at random.
		\item For $i = 0$ to $r - 1$, run the following \emph{round consistency test}:
		\begin{enumerate}
			\item Define $Q_{i + 1} \in \calP_{i + 1}$ by $Q_{i + 1} = \pi_i(Q_i)$,
			\item Query $f^{(i+1)}$ to get $f^{(i+1)}(Q_{i+1})$ and query $f^{(i)}$ at points $\hat Q \in S_{Q_{i+1}}$,
			\item[] (if $i = r -1$, set $f^{(r)}(Q_{r}) = \beta$)
			\item Compute the value $\fold{\xpi{f}, \bfz^{(i)}}(Q_{i+1})$,
			\item If $i < r - 1$: return \textsf{\textbf{reject}} if and only if $f^{(i+1)}(Q_{i+1}) \neq \fold{f^{(i)}, \bfz^{(i)}}(Q_{i+1})$
			\item If $i = r - 1$: return \textsf{\textbf{reject}} if and only if $\beta \neq \fold{f^{(i)}, \bfz^{(i)}}(Q_{i+1})$
		\end{enumerate}
	\end{enumerate}
	\item Return \textsf{\textbf{acccept}}. 
\end{enumerate}

\subsection{Properties of the AG-IOPP with Kummer curves}

Assume $C_0 = C(\calX_0, \calP_0, D_0)$ is a foldable AG code of blocklength $n_0 = \size{\calP_0}$ on a Kummer curve $\calX_0$ (cf. Proposition \ref{prop:Kummer_foldable}). This means that $\calX_0$ is defined by an equation $y^N = f(x)$, where $f \in \F[X]$ is a separable degree-$m$ polynomial, $m \equiv -1 \mod N$, $N$ is coprime with $\size{\F}$, $\size{\calP_0} = \alpha N$ for some integer $\alpha$, and $\deg D_0 < \alpha N$. Assume $\alpha$ is a power of 2 and $N$ is a $\eta$-smooth integer for a small fixed parameter $\eta \in \mathbb N$.

Proposition \ref{prop:dist_min_Kummer} states that the relative minimum distances of the codes $C_i$ are all equal to $\Delta(C_0) = 1 - \frac{\deg D_0}{\alpha N}$. Therefore, the ordering on the integers involved in the prime decomposition $\prod_{i = 0}^{s-1}p_i$ of $N$ does not impact the parameters of the protocol. Moreover, the code $C_{s} = C(\calX_{s}, \calP_{s}, D_{s})$ corresponds to a RS code \[
C_{s} = \RS{\calP_{s}, \frac{\deg D_0}{N}} = \set{f : \calP_{s} \rightarrow \F ; \deg f \le \frac{\deg D_0}{N}}
\] of blocklength $\size{\calP_s} = \alpha$, which is itself a foldable AG code (see Example \ref{ex:RS-foldable}). 

\begin{theorem}[Kummer case]\label{thm:kummer-properties}
	Let $C= (\calX_0, \calP_0, D_0)$ be a foldable AG code on a Kummer curve satisfying the hypotheses of Proposition \ref{prop:Kummer_foldable} with $N$ a $\eta$-smooth integer. Denote $n = \size{\calP_0}$.
	The IOPP $(\prover, \verifier)$ described in Section \ref{subsec:fast-iopp} has perfect completeness and soundness as stated in Theorem \ref{thm:properties}. Moreover, for $t$ repetitions of the \textsf{QUERY} phase, we have:
	\[
	\begin{array}{lll}
		\text{rounds complexity} &\mathsf{r}(n) &< \log n,\\
		\text{proof length} &\mathsf{l}(n) &< n,\\
		\text{query complexity} &\mathsf{q}(n) &\le t \eta \log_2 n + 1,\\
		\text{prover complexity} &\mathsf{t_p}(n) &= O_{\eta}(n),\\
		\text{verifier decision complexity} & \mathsf{t_v}(n) &= O_{\eta}(t\log n).
	\end{array}
	\]
\end{theorem}

\begin{proof}   
	Noticing that the round complexity is now $\mathsf{r}(n) = s + s'$, straightforward calculations show that complexity, query complexity and proof length computed in the proof \ref{thm:properties} still hold. Completeness and soundness also follow from \cite{ABN21}.
	We estimate prover complexity and verifier complexity below.
	
	\proofsubparagraph*{Prover complexity.}
	Fix a round index $i < r - 1$. The balancing functions $\nu_{i+1, j} : \calP_{i+1} \rightarrow \F$ can be precomputed since they do not depend on $f^{(i)}, \bfz^{(i)}$ (see Remark \ref{rem:cost-balancing}).       
	To simplify notation, denote $f = \xpi{f}$. For any $\bfz = (z_1, z_2) \in \F^2$, computing the successive powers $(z_1^j, z_2^j)_{0 \le j < p_i}$ takes $2(p_i - 2)$ multiplications.         
	For each $P \in \calP_{i+1}$, an honest prover must compute the coefficients $(a_{j,P})_{0 \le j, < P}$ of the polynomial $I_{f, P}(X)$ of degree $\deg I_{f,P} < p_i$ from the interpolation set $\set{(\mu_i(\hat P), f(\hat P)) \mid \hat P \in S_{P}}$ of size $p_i$. Notice that $\mu_i = y$, so computing $\mu_i(\hat P)$ for $\hat P \in S_P$ is done for free. Univariate interpolation for a polynomial of degree $<p_i$ can be done in $O(p_i^2)$ by Lagrange interpolation.    
	Overall, one can honestly evaluate $\fold{f, \bfz} : \calP_{i + 1} \rightarrow \F$ with $\size{\calP_{i+1}}O(p_i^2)$ operations in $\F$. We showed previously that $\sum_{i=1}^{r - 1}\size{\calP_i} < n$, thus when summing over $r - 1$ rounds, we get that the cost of (honestly) generating the oracles $f^{(1)}, \ldots, f^{(r - 1)}$ is $O_{\eta }(n)$.        
	
	\proofsubparagraph*{Verifier decision complexity.} 
	
	Verifier complexity is inferred from the previous discussion about prover complexity. For each round, the verifier computes the successive powers of $z_1$ and $z_2$, interpolates $I_{f,P}$ for a point $P \in \calP_{i+1}$ in $O(p_i^2)$ operations, then computes $\fold{f, \bfz}(P)$ in a number of operations which is independent of $n$. Hence, verifier complexity for repetition parameter $t$ is $\mathsf{t_v}(n) = O_{\eta}(t\log(n))$.   
\end{proof}

\begin{remark}\label{rem:cost-balancing}
	We give the cost of precomputing the evaluation tables of the balancing functions. Letting $\nu_{i+1, j}$ be as defined in proof of Lemma \ref{lem:compatible-Kummer}, the sequence of functions $(\nu_{i+1, j})_{0 < j < p_i}$ can be evaluated at the same point $P \in \calP_{i + 1}$ in time $O(\log m + p_i)$ using exponentiation by squaring. Thus, the evaluations of $\nu_{i+1, 1}, \dots \nu_{i+1, p_i - 1}$ on $\calP_{i+1}$ are obtained with $O((\log m + p_i)\size{\calP_{i+1}})$ operations.
\end{remark}

We give an example of an AG code over a Kummer curve where $p_{max} = 2$.

\begin{example} On $\F_{q^2}$ with $q=2^{61}-1$ ($9^{th}$ Mersenne prime), we consider the curve
	\[\calX_0 : y^N=x^3+x\]
	where $N=2^{r}$ with $r=16$. It is maximal \cite{TT14} of genus $g=N-1$. We consider the code $C_0$ associated to $D_0= 2^{17} P_\infty^0$ on an evaluation set $\mathcal{P}_0 \subset \calX_0(\F_{q^2})$ of size $n=2^{20}$. Its dimension equals $\dim C_0=2^{16}+2$ and its relative minimum distance $\lambda$ is bounded from below by $1- 2^{-3}$. Take $\epsilon = 2 ^{-6.55}$. By Theorem \ref{thm:properties},
	\[\err_{commit} \leq \frac{\log(n)}{\size{\F_{q^2}}}\left(1+\frac{4}{\epsilon}\right) \left(\frac{4}{\epsilon}\right)^2 \approx 2^{4.33+6+3\cd 6.55 -121} \leq 2^{-91}\]

	\[\err_{query}(\delta) \leq \left(1-\delta+ \epsilon\log(n)\right)\]
	where $1-\delta=(1-\lambda + \epsilon)^{\frac{1}{3}} \leq 0,51384$. Hence
	\[\err_{query}(\delta) \leq 0,51384+ \frac{20}{2^{6.55}} \approx 0,72728.\]
	By running the \textsf{QUERY} phase with repetition parameter $t \geq 199$, we get $(\err_{query})^t \leq 2^{-91}$ and $\err(\delta)\leq 2^{-90}.$
	The last code $C_r$ is a small Reed-Solomon code of length $n_r = 2^{4}$ and dimension $2$. The total number of rounds of the IOPP is thus $r + 1$.
\end{example}

\subsection{Properties of the AG-IOPP with towers of Hermitian curves}

\begin{theorem}\label{thm:tower-properties}
	Let $C=(\calX, \calP, D)$ be a foldable AG code with alphabet $\F = \F_{q^2}$ on a tower of Hermitian curves satisfying the hypotheses of Proposition \ref{prop:tower-family}. Letting $\ii$ be the index of the curve $\calX$ in the Hermitian tower $(\calX_i)_{i \ge 0}$, the length $n = \size \calP$ of $C$ is at most $q^{\ii + 2}$. 
	The IOPP $(\prover, \verifier)$ described in Section \ref{subsec:fast-iopp} has perfect completeness, and soundness as stated in Theorem \ref{thm:properties}. Moreover, we have:
	\[
	\begin{array}{lll}
		\text{rounds complexity} &\mathsf{r}(n) &< \log n,\\
		\text{proof length} &\mathsf{l}(n) &< n,\\
		\text{query complexity} &\mathsf{q}(n) &\le t q \log n + 1,
		\\
		\text{prover complexity} &\mathsf{t_p}(n) &= O(n \cdot \mathsf{M}_\F(q)\log(q)),\\
		\text{verifier decision complexity} & \mathsf{t_v}(n) &= O(\log n \cdot\mathsf{M}_\F(q)\log(q)).
	\end{array}
	\]
\end{theorem}

\begin{proof}
	The proof follows from proof of Theorem \ref{thm:kummer-properties}, replacing $\eta$ by $q$. Prover and verifier complexities are computed from the cost of computing the coefficients of a univariate polynomial of degree less than $q$ from its evaluation on points forming an arithmetic progression in $\F = \F_{q^2}$. This interpolation task can be done in $\mathsf{M}_\F(q)\log q + O(\mathsf{M}_\F(q))$ base field operations \cite{BrS05}, where $\mathsf{M}_\F(d)$ denotes the cost of multiplying two degree-$d$ univariate polynomials in $\F[X]$.
\end{proof}

Given a foldable code as in Proposition \ref{prop:tower-family}, the IOPP constructs a sequence of codes as follow:
\[C_i \mydef C(\calX_{\ii-i},\calP_{\ii-i},D_{\ii-i}) \text{ where } \calP_{i-1}=\pi_i(\calP_i)\text{ and } D_i = d_iP^{(i)}_{\infty} \]
with the integers $d_i$ defined recursively by
\[
d_{i-1} \mydef \floor{ \dfrac{d_i}{q}} + 2g(\calX_{i-1}).
\]
\begin{remark}
	Be careful about the indices: the indices for the codes on one hand and for the curves, the support and the degree on the other hand, are in reserved order. The original code to which we test proximity is $C_0$, defined over the curve $\calX_{\ii}$, and the last RS code is $C_{\ii}$, which is defined over $\calX_0\simeq \PP^1$.
\end{remark}

Unlike the Kummer case, we have to increase the degree of divisor by twice the genus of the curve at each step to make sure the compatibility hypotheses of Definition \ref{def-div_compatible} are valid. This has a counterpart: the dimension of the codes $C_i$ decreases much slowly than their block length. A foldable code in the sense of Definition \ref{def:good_properties} may induce of a sequence of codes in which the last code $C_{\ii}$ is trivial. In this case, the protocol would no longer be sound. We thus need to control the dimension of the code $C_{\ii}$. This is the purpose of the remaining of this section.

\begin{remark}
	In light of the Kummer case in which the group $\Z/N\Z$ is factored as much as possible, if $q$ is some prime power $q=p^\ell$, we could split the additive group $\Fq \simeq \left(\Z/pZ\right)^\ell$ acting at each level with $\ell$ intermediary rounds. In \ref{lem:locality}, any value taken by the folding of a function $f$ would be determined by $p$ values of $f$ (instead of $q$ values), and the verifier and the prover would perform polynomial interpolations of degree $p$. However, for each of these intermediary steps, we would have to make the new divisor grow to fulfill the compatibility conditions (\autoref{def-div_compatible}), as mentioned above. If the rate increases too much, the relative minimum distance drops and the total number queries to target a designated soundness may be tremendous. The loss in terms of soundness error per \textsf{QUERY} phase seems to be much more significant than the aforementioned advantages. 
\end{remark}

\subsubsection{Bounding the rate of the underlying Reed-Solomon code}

We aim to bound the dimension of the code Reed-Solomon code $C_{\ii}$. Let us compute the degree $d_{\ii}$ of the divisor $D_{\ii}$ on $\mathbb{P}^1$.

\begin{lemma}\label{lem:deg_herm}
	For $1 \leq j \leq \ii$, we have 
	\[d_{\ii -j} \leq \left\lfloor \dfrac{d_{\ii}}{q^j}\right\rfloor + \sum\limits_{k=1}^j \left\lfloor \dfrac{2g_{\ii -k}}{q^{j-k}}\right\rfloor + (j-1).\] 
\end{lemma}

\begin{proof}
	It follows from the definition of the degrees $d_i$ given in \eqref{eq:deg_herm} and by induction on $j$.
\end{proof}

From Lemma \ref{lem:deg_herm}, we can get an upper bound on $d_0$.

\begin{corollary}
	Let us assume that $2(\ii-1) < q$.
	The degree $d_0$ of the divisor $D_0$ on $\mathbb{P}^1$ is bounded from above by 
	\[d_0 \leq \left\lfloor \dfrac{d_{\ii}}{q^{\ii}}\right\rfloor + (\ii-1)\left( 1+\frac{\ii}{6} \cd \left(3q-4 + 2\ii\right) \right) \]
\end{corollary}

\begin{proof}
	By Lemma \ref{lem:deg_herm}, we have the following bound over $d_0$:
	\[d_0 \leq \left\lfloor \dfrac{d_{\ii}}{q^{\ii}}\right\rfloor + \sum\limits_{i=0}^{\ii-1} \left\lfloor \dfrac{2g_i}{q^{i}}\right\rfloor + \ii-1.\]
	It is thus enough to estimate the sum $\sum\limits_{k=0}^{\ii-1} \left\lfloor \dfrac{2g_k}{q^k}\right\rfloor$. By Proposition \ref{prop:maj_genus_herm},
	\begin{align*}
		\sum\limits_{k=0}^{\ii-1} \left\lfloor \dfrac{2g_k}{q^k}\right\rfloor & \leq  \sum\limits_{k=0}^{\ii-1} \left(kq + k(k-1)\right)  \\
		&= (q-1) \cd \dfrac{\ii(\ii-1)}{2} + \dfrac{\ii(\ii-1)(2\ii-1)}{6} \\
		&= \dfrac{\ii(\ii-1)}{2}\cd \left(q-\frac{4}{3} + \frac{2\ii}{3}\right),
	\end{align*}
	which gives the expected result.
\end{proof}

We aim to determine a sufficient condition on $\ii$ that ensures that the RS code $C_{\ii}$ is not trivial. Let us denote by $n_0$ the size of the support $\calP_0$ of $C_{\ii}$. It satisfies $n_0 \leq q^2$. The rate of $C_ {\ii}$ is equal to \[\frac{d_0+1}{n_0}.\] Also we have $n_{\ii} := \size{\calP_{\ii}} = q^{\ii}n_0$.

\begin{corollary}\label{cor:cdt_deg_herm}
	Let us fix $\rho \in (0,1)$. If
	\[ \left\lfloor \dfrac{d_{\ii}}{q^{\ii}}\right\rfloor +  (\ii-1)\left( 1+\frac{\ii}{6} \cd \left(3q-4 + 2\ii\right)\right) +1 < \rho n_0,\]
	then the rate of the RS code $C_{\ii}$ is less than $\rho$.
\end{corollary}

\subsubsection{Foldable codes with constant rate which are endowed with an IOPP with designed soundness}
In this paragraph, we focus on foldable codes of the form 
\begin{equation}\tag{\ref{eq:def_Cimax}}
	C_0= C\left(\calX_{\ii},\calX_{\ii}(\F_{q^2})\setminus \{P^{(\ii)}_{\infty} , (2\alpha+1)g_{\ii})P^{(\ii)}_{\infty}\right)
\end{equation}
for some $\alpha > 1/2$, as in Section \ref{subsec:foldable-towers-codes}. The evaluation set $\calP_{\ii}$ is the whole set of rational points of $\calX_{\ii}$ minus the point of at infinity $P_\infty^{(\ii)}$, \textit{i.e.} $n_{\ii}=q^{\ii +2}$.

\begin{proposition}
	Let us fix $\rho \in (0,1)$. The rate of the RS code $C_{\ii}$ below $C_0$ is less than $\rho$ if 
	\[2\ii^3 + 3\ii^2(2\alpha+q-1) + \ii (6\alpha(q-1)+7) - 6\rho q^2 <0.\]
\end{proposition}

\begin{proof}
	For $q$ large enough, we can assume that $2\ii-1<q$. Using Proposition \ref{prop:maj_genus_herm}, we get an upper bound over $d_{\ii}$: 
	\[ 
	d_{\ii} \leq (2\alpha+1)\dfrac{\ii}{2}q^{\ii}(q+(\ii-1)).
	\]
	From Corollary \ref{cor:cdt_deg_herm}, a sufficient condition for the underlying RS code to have a rate less than $\rho$ is    
	\[
	(2\alpha+1)\dfrac{\ii}{2}(q+(\ii-1)) + (\ii-1)\left(1+\frac{\ii}{6}(3q-4+2\ii)\right)+1 < \rho q^2
	\]    
	Multiplying the inequality by $6$, expanding and simplifying, we get our condition.
\end{proof}

Now assume that $\ii=q^{\varepsilon}$ for $\varepsilon \in (0,1)$. In the constant rate regime described in Lemma \ref{lem:limit-rate}, we have $\alpha \ii = Rq$. The condition above becomes 
\[2\ii^3 + 3\ii^2(q-1) + \ii(6Rq+7)< 6q^2\left(\rho-R\left(1-\frac{1}{q}\right)\right) .\]

If $R\left(1-\frac{1}{q}\right) < \rho$, the right handside is positive. Let us give a rough estimation of the largest $\epsilon$ such that $\ii=q^\epsilon$ satisfies this inequality. The left handside begin equivalent to $3q^{1+2\epsilon}$, we have $\epsilon \simeq \frac{1}{2}(1+\log_q\left(\rho-R\left(1-\frac{1}{q}\right)\right)$.

Table \ref{table:param-HT} displays some examples of level $\ii$ and initial rate $R$ of foldable codes for which the AG-IOPP reduce the proximity test to testing RS codes of rate $\rho$. In terms of the soundness of the protocol, it means that $\lambda$ as defined in Theorem \ref{thm:properties} is greater than $1-\rho$.

\begin{table}[H]
	\centering
	\renewcommand{\arraystretch}{1.2}
	\begin{tabular}{|c|c|c|c|c|}
		\hline
		$q$      & $\ii$ & $n$  & $R$    & $1-\rho>$   \\\hline
		$2^4$      & 3  & $2^{20}$ & \multirow{2}{*}{$1/8$} & \multirow{2}{*}{$1/3$} \\ \cline{1-3}
		$2^5$      & 5  & $2^{35}$ &         &          \\\hline
		$2^4$      & 4  & $2^{24}$ & \multirow{6}{*}{$1/16$} & $1/3$ \\ \cline{1-3} \cline{5-5}
		\multirow{2}{*}{$2^5$}  & 3  & $2^{25}$ &  & $3/4$ \\ \cline{2-3} \cline{5-5}
		& 5  & $2^{35}$ &  & $1/2$ \\ \cline{1-3} \cline{5-5}
		\multirow{3}{*}{$2^6$}      & 4  & $2^{36}$ &  & $3/4$ \\ \cline{2-3} \cline{5-5}
		& 5  & $2^{42}$ &  & $2/3$ \\ \cline{2-3} \cline{5-5}
		& 7  & $2^{54}$ &  & $1/2$ \\ \hline
		$2^4$      & 3  & $2^{20}$ & 1/32 & $1/2$ \\ \hline
	\end{tabular}
	\renewcommand{\arraystretch}{1}
	\caption{Example of parameters of foldable codes of rate $R$ along the Hermitian tower. Alphabet is $\F_q^2$ and block length is $n$. The last column gives a bound on the minimal distance of the RS code.}\label{table:param-HT}
\end{table}



\bibliography{biblio-final}

\begin{thebibliography}{10}

\bibitem{AHIV17}
Scott Ames, Carmit Hazay, Yuval Ishai, and Muthuramakrishnan
  Venkitasubramaniam.
\newblock {L}igero: {L}ightweight {S}ublinear {A}rguments {W}ithout a {T}rusted
  {S}etup.
\newblock In Bhavani~M. Thuraisingham, David Evans, Tal Malkin, and Dongyan Xu,
  editors, {\em Proceedings of the 2017 {ACM} {SIGSAC} Conference on Computer
  and Communications Security, {CCS} 2017, Dallas, TX, USA, October 30 -
  November 03, 2017}, pages 2087--2104. {ACM}, 2017.
\newblock \href {https://doi.org/10.1145/3133956.3134104}
  {\path{doi:10.1145/3133956.3134104}}.

\bibitem{ALM+98}
Sanjeev Arora, Carsten Lund, Rajeev Motwani, Madhu Sudan, and Mario Szegedy.
\newblock {P}roof {V}erification and the {H}ardness of {A}pproximation
  {P}roblems.
\newblock 45(3):501--555, 1998.
\newblock extended version of FOCS'92.
\newblock \href {https://doi.org/10.1145/278298.278306}
  {\path{doi:10.1145/278298.278306}}.

\bibitem{AS92}
Sanjeev Arora and Shmuel Safra.
\newblock {P}robabilistic {C}hecking of {P}roofs; {{A}} {N}ew
  {C}haracterization of {{N}{P}}.
\newblock In {\em 33rd Annual Symposium on Foundations of Computer Science,
  Pittsburgh, Pennsylvania, USA, 24-27 October 1992}, pages 2--13. {IEEE}
  Computer Society, 1992.
\newblock \href {https://doi.org/10.1145/273865.273901}
  {\path{doi:10.1145/273865.273901}}.

\bibitem{ABN21}
Daniel Augot, Sarah Bordage, and Jade Nardi.
\newblock Efficient multivariate low-degree tests via interactive oracle proofs
  of proximity for polynomial codes.
\newblock {\em Electron. Colloquium Comput. Complex.}, page 118, 2021.
\newblock URL: \url{https://eccc.weizmann.ac.il/report/2021/118}.

\bibitem{Bab85}
L{\'{a}}szl{\'{o}} Babai.
\newblock {T}rading {G}roup {T}heory for {R}andomness.
\newblock In Robert Sedgewick, editor, {\em Proceedings of the 17th Annual
  {ACM} Symposium on Theory of Computing, May 6-8, 1985, Providence, Rhode
  Island, {USA}}, pages 421--429. {ACM}, 1985.
\newblock \href {https://doi.org/10.1145/22145.22192}
  {\path{doi:10.1145/22145.22192}}.

\bibitem{BFLS91}
L{\'{a}}szl{\'{o}} Babai, Lance Fortnow, Leonid~A. Levin, and Mario Szegedy.
\newblock {C}hecking {C}omputations in {P}olylogarithmic {T}ime.
\newblock In {\em Proceedings of the 23rd Annual {ACM} Symposium on Theory of
  Computing, May 5-8, 1991, New Orleans, Louisiana, {USA}}, pages 21--31, 1991.
\newblock \href {https://doi.org/10.1145/103418.103428}
  {\path{doi:10.1145/103418.103428}}.

\bibitem{BBGS14}
Alp {Bassa}, Peter {Beelen}, Arnaldo {Garcia}, and Henning {Stichtenoth}.
\newblock An {I}mprovement of the {G}ilbert–{V}arshamov {B}ound {O}ver
  {N}onprime {F}ields.
\newblock {\em {IEEE} {T}ransactions on {I}nformation {T}heory},
  60(7):3859--3861, 2014.

\bibitem{BRS20}
Peter Beelen, Johan Rosenkilde, and Grigory Solomatov.
\newblock Fast encoding of {AG} codes over c\({}_{\mbox{ab}}\) curves.
\newblock {\em {IEEE} Trans. Inf. Theory}, 67(3):1641--1655, 2021.
\newblock \href {https://doi.org/10.1109/TIT.2020.3042248}
  {\path{doi:10.1109/TIT.2020.3042248}}.

\bibitem{BBHR18a}
Eli Ben{-}Sasson, Iddo Bentov, Yinon Horesh, and Michael Riabzev.
\newblock {F}ast {R}eed-{S}olomon {I}nteractive {O}racle {P}roofs of
  {P}roximity.
\newblock In {\em 45th International Colloquium on Automata, Languages, and
  Programming, {ICALP} 2018, July 9-13, 2018, Prague, Czech Republic}, pages
  14:1--14:17, 2018.

\bibitem{BBHR19}
Eli Ben{-}Sasson, Iddo Bentov, Yinon Horesh, and Michael Riabzev.
\newblock {S}calable {Z}ero {K}nowledge with {N}o {T}rusted {S}etup.
\newblock In Alexandra Boldyreva and Daniele Micciancio, editors, {\em Advances
  in Cryptology - {CRYPTO} 2019 - 39th Annual International Cryptology
  Conference, Santa Barbara, CA, USA, August 18-22, 2019, Proceedings, Part
  {III}}, volume 11694 of {\em Lecture Notes in Computer Science}, pages
  701--732. Springer, 2019.
\newblock \href {https://doi.org/10.1007/978-3-030-26954-8_23}
  {\path{doi:10.1007/978-3-030-26954-8_23}}.

\bibitem{BCIKS20}
Eli Ben{-}Sasson, Dan Carmon, Yuval Ishai, Swastik Kopparty, and Shubhangi
  Saraf.
\newblock {P}roximity {G}aps for {R}eed-{S}olomon {C}odes.
\newblock {\em {IACR} Cryptol. ePrint Arch.}, 2020:654, 2020.

\bibitem{BCGRS17}
Eli Ben{-}Sasson, Alessandro Chiesa, Ariel Gabizon, Michael Riabzev, and
  Nicholas Spooner.
\newblock {I}nteractive {O}racle {P}roofs with {C}onstant {R}ate and {Q}uery
  {C}omplexity.
\newblock In {\em 44th International Colloquium on Automata, Languages, and
  Programming, {ICALP} 2017, July 10-14, 2017, Warsaw, Poland}, pages
  40:1--40:15, 2017.

\bibitem{BCGGRS19}
Eli Ben{-}Sasson, Alessandro Chiesa, Lior Goldberg, Tom Gur, Michael Riabzev,
  and Nicholas Spooner.
\newblock {L}inear-{S}ize {C}onstant-{Q}uery {I}{O}{P}s for {D}elegating
  {C}omputation.
\newblock In Dennis Hofheinz and Alon Rosen, editors, {\em Theory of
  Cryptography - 17th International Conference, {TCC} 2019, Nuremberg, Germany,
  December 1-5, 2019, Proceedings, Part {II}}, volume 11892 of {\em Lecture
  Notes in Computer Science}, pages 494--521. Springer, 2019.

\bibitem{BCRSVW19}
Eli Ben{-}Sasson, Alessandro Chiesa, Michael Riabzev, Nicholas Spooner, Madars
  Virza, and Nicholas~P. Ward.
\newblock {A}urora: {T}ransparent {S}uccinct {A}rguments for {{R}1{C}{S}}.
\newblock In Yuval Ishai and Vincent Rijmen, editors, {\em Advances in
  Cryptology - {EUROCRYPT} 2019 - 38th Annual International Conference on the
  Theory and Applications of Cryptographic Techniques, Darmstadt, Germany, May
  19-23, 2019, Proceedings, Part {I}}, volume 11476 of {\em Lecture Notes in
  Computer Science}, pages 103--128. Springer, 2019.

\bibitem{BCS16}
Eli Ben{-}Sasson, Alessandro Chiesa, and Nicholas Spooner.
\newblock {I}nteractive {O}racle {P}roofs.
\newblock In {\em Theory of Cryptography - 14th International Conference, {TCC}
  2016-B, Beijing, China, October 31 - November 3, 2016, Proceedings, Part
  {II}}, pages 31--60, 2016.

\bibitem{BGKS20}
Eli Ben{-}Sasson, Lior Goldberg, Swastik Kopparty, and Shubhangi Saraf.
\newblock {{D}{E}{E}{P}-{F}{R}{I}:} {S}ampling {O}utside the {B}ox {I}mproves
  {S}oundness.
\newblock In {\em 11th Innovations in Theoretical Computer Science Conference,
  {ITCS} 2020, January 12-14, 2020, Seattle, Washington, {USA}}, pages
  5:1--5:32, 2020.

\bibitem{BKKMS13}
Eli Ben{-}Sasson, Yohay Kaplan, Swastik Kopparty, Or~Meir, and Henning
  Stichtenoth.
\newblock {C}onstant {R}ate {{P}{C}{P}}s for {C}ircuit-{S}{A}{T} with
  {S}ublinear {Q}uery {C}omplexity.
\newblock In {\em 54th Annual {IEEE} Symposium on Foundations of Computer
  Science, {FOCS} 2013, 26-29 October, 2013, Berkeley, CA, {USA}}, pages
  320--329. {IEEE} Computer Society, 2013.

\bibitem{BKS18}
Eli Ben{-}Sasson, Swastik Kopparty, and Shubhangi Saraf.
\newblock {W}orst-{C}ase to {A}verage {C}ase {R}eductions for the {D}istance to
  a {C}ode.
\newblock In {\em 33rd Computational Complexity Conference, {CCC} 2018, June
  22-24, 2018, San Diego, CA, {USA}}, pages 24:1--24:23, 2018.

\bibitem{BS08}
Eli Ben{-}Sasson and Madhu Sudan.
\newblock {S}hort {{P}{C}{P}}s with {P}olylog {Q}uery {C}omplexity.
\newblock {\em {SIAM} J. Comput.}, 38(2):551--607, 2008.

\bibitem{BCL20}
Jonathan Bootle, Alessandro Chiesa, and Siqi Liu.
\newblock {Z}ero-{K}nowledge {S}uccinct {A}rguments with a {L}inear-{T}ime
  {P}rover.
\newblock {\em {IACR} Cryptol. ePrint Arch.}, page 1527, 2020.
\newblock URL: \url{https://eprint.iacr.org/2020/1527}.

\bibitem{BrS05}
Alin {Bostan} and Eric {Schost}.
\newblock {P}olynomial evaluation and interpolation on special sets of points.
\newblock {\em Journal of Complexity}, 21(4):420--446, 2005.
\newblock \href {https://doi.org/10.1016/j.jco.2004.09.009}
  {\path{doi:10.1016/j.jco.2004.09.009}}.

\bibitem{COS20}
Alessandro Chiesa, Dev Ojha, and Nicholas Spooner.
\newblock {F}ractal: {P}ost-quantum and {T}ransparent {R}ecursive {P}roofs from
  {H}olography.
\newblock In Anne Canteaut and Yuval Ishai, editors, {\em Advances in
  Cryptology - {EUROCRYPT} 2020 - 39th Annual International Conference on the
  Theory and Applications of Cryptographic Techniques, Zagreb, Croatia, May
  10-14, 2020, Proceedings, Part {I}}, volume 12105 of {\em Lecture Notes in
  Computer Science}, pages 769--793. Springer, 2020.

\bibitem{Dinur07}
Irit Dinur.
\newblock {T}he {{P}{C}{P}} theorem by gap amplification.
\newblock {\em J. {ACM}}, 54(3):12, 2007.

\bibitem{GMR85}
Shafi Goldwasser, Silvio Micali, and Charles Rackoff.
\newblock {T}he {K}nowledge {C}omplexity of {I}nteractive {P}roof-{S}ystems
  ({E}xtended {A}bstract).
\newblock In Robert Sedgewick, editor, {\em Proceedings of the 17th Annual
  {ACM} Symposium on Theory of Computing, May 6-8, 1985, Providence, Rhode
  Island, {USA}}, pages 291--304. {ACM}, 1985.

\bibitem{Goppa77}
Valerii~Denisovich Goppa.
\newblock Codes associated with divisors.
\newblock {\em Problemy Peredachi Informatsii}, 13(1):33--39, 1977.

\bibitem{HKT13}
J.~W.~P. Hirschfeld, G.~Korchmáros, and F.~Torres.
\newblock {\em {A}lgebraic {C}urves over a {F}inite {F}ield}.
\newblock Princeton University Press, Princeton, 25 Mar. 2013.
\newblock URL: \url{https://www.degruyter.com/princetonup/view/title/511887},
  \href {https://doi.org/10.1515/9781400847419}
  {\path{doi:10.1515/9781400847419}}.

\bibitem{HY18}
Chuangqiang {Hu} and Shudi {Yang}.
\newblock {M}ulti-point codes over {K}ummer extensions.
\newblock {\em Designs, {C}odes and {C}ryptography}, 86:211--230, 2018.

\bibitem{KR08}
Yael~Tauman Kalai and Ran Raz.
\newblock {I}nteractive {{P}{C}{P}}.
\newblock In Luca Aceto, Ivan Damg{\aa}rd, Leslie~Ann Goldberg, Magn{\'{u}}s~M.
  Halld{\'{o}}rsson, Anna Ing{\'{o}}lfsd{\'{o}}ttir, and Igor Walukiewicz,
  editors, {\em Automata, Languages and Programming, 35th International
  Colloquium, {ICALP} 2008, Reykjavik, Iceland, July 7-11, 2008, Proceedings,
  Part {II} - Track {B:} Logic, Semantics, and Theory of Programming {\&} Track
  {C:} Security and Cryptography Foundations}, volume 5126 of {\em Lecture
  Notes in Computer Science}, pages 536--547. Springer, 2008.

\bibitem{Kani}
Ernst Kani.
\newblock {T}he {G}alois-module structure of the space of holomorphic
  differentials of a curve.
\newblock {\em Journal für die reine und angewandte Mathematik}, 367:187--206,
  1986.
\newblock URL: \url{http://eudml.org/doc/152831}, \href
  {https://doi.org/10.1515/crll.1986.367.187}
  {\path{doi:10.1515/crll.1986.367.187}}.

\bibitem{KPV19}
Assimakis Kattis, Konstantin Panarin, and Alexander Vlasov.
\newblock {R}ed{S}hift: {T}ransparent {S}{N}{A}{R}{K}s from {L}ist {P}olynomial
  {C}ommitment {I}{O}{P}s.
\newblock {\em {IACR} Cryptol. ePrint Arch.}, 2019:1400, 2019.
\newblock URL: \url{https://eprint.iacr.org/2019/1400}.

\bibitem{Kil92}
Joe Kilian.
\newblock {A} {N}ote on {E}fficient {Z}ero-{K}nowledge {P}roofs and {A}rguments
  ({E}xtended {A}bstract).
\newblock In S.~Rao Kosaraju, Mike Fellows, Avi Wigderson, and John~A. Ellis,
  editors, {\em Proceedings of the 24th Annual {ACM} Symposium on Theory of
  Computing, May 4-6, 1992, Victoria, British Columbia, Canada}, pages
  723--732. {ACM}, 1992.

\bibitem{L87}
{Gilles} {L}achaud.
\newblock Sommes d'{E}isenstein et nombre de points de certaines courbes
  algébriques sur les corps finis.
\newblock {\em C. R. Acad. Sci. Paris}, 305, 01 1987.

\bibitem{L92}
Gilles {Lachaud}.
\newblock {A}rtin-{S}chreier curves, exponential sums, and coding theory.
\newblock {\em {T}heoretical {C}omputer {S}cience}, 94(2):295--310, 1992.
\newblock URL:
  \url{https://www.sciencedirect.com/science/article/pii/030439759290040M},
  \href {https://doi.org/https://doi.org/10.1016/0304-3975(92)90040-M}
  {\path{doi:https://doi.org/10.1016/0304-3975(92)90040-M}}.

\bibitem{LFKN90}
Carsten Lund, Lance Fortnow, Howard~J. Karloff, and Noam Nisan.
\newblock {A}lgebraic {M}ethods for {I}nteractive {P}roof {S}ystems.
\newblock In {\em 31st Annual Symposium on Foundations of Computer Science, St.
  Louis, Missouri, USA, October 22-24, 1990, Volume {I}}, pages 2--10. {IEEE}
  Computer Society, 1990.

\bibitem{M04}
Hiren Maharaj.
\newblock {C}ode {C}onstruction on {F}iber {P}roducts of {K}ummer {C}overs.
\newblock {\em Information Theory, IEEE Transactions on}, 50:2169 -- 2173, 10
  2004.
\newblock \href {https://doi.org/10.1109/TIT.2004.833356}
  {\path{doi:10.1109/TIT.2004.833356}}.

\bibitem{MQS15}
Ariane~M. {Masuda}, Luciane {Quoos}, and Alonso {Sep{\'u}lveda}.
\newblock {{O}ne- and {T}wo-{P}oint {C}odes over {K}ummer {E}xtensions}.
\newblock {\em {IEEE} Transactions on Information Theory}, 62(9):4867--4872,
  2016.
\newblock \href {https://doi.org/10.1109/TIT.2016.2583437}
  {\path{doi:10.1109/TIT.2016.2583437}}.

\bibitem{Mei13}
Or~Meir.
\newblock {{I}{P}} = {{P}{S}{P}{A}{C}{E}} {U}sing {E}rror-{C}orrecting {C}odes.
\newblock {\em {SIAM} J. Comput.}, 42(1):380--403, 2013.

\bibitem{Mic95}
Silvio Micali.
\newblock {C}omputationally-{S}ound {P}roofs.
\newblock In Johann~A. Makowsky and Elena~V. Ravve, editors, {\em Proceedings
  of the Annual European Summer Meeting of the Association of Symbolic Logic,
  Logic Colloquium 1995, Haifa, Israel, August 9-18, 1995}, volume~11 of {\em
  Lecture Notes in Logic}, pages 214--268. Springer, 1995.

\bibitem{Mie09}
Thilo Mie.
\newblock {S}hort {P}{C}{P}{P}s {V}erifiable in {P}olylogarithmic {T}ime with
  {O}(1) {Q}ueries.
\newblock {\em Annals of Mathematics and Artificial Intelligence},
  56(3–4):313–338, August 2009.

\bibitem{MP93}
Carlos Munuera and Ruud Pellikaan.
\newblock {E}quality of geometric {G}oppa codes and equivalence of divisors.
\newblock {\em Journal of Pure and Applied Algebra}, 90(3):229 -- 252, 1993.
\newblock \href {https://doi.org/10.1016/0022-4049(93)90043-S}
  {\path{doi:10.1016/0022-4049(93)90043-S}}.

\bibitem{PSW91}
Ruud {Pellikaan}, Ba~zhong {Shen}, and Gerhard {J. M. van Wee}.
\newblock Which linear codes are algebraic-geometric?
\newblock {\em IEEE Trans. Inf. Theory}, 37:583--602, 1991.

\bibitem{RRR16}
Omer Reingold, Guy~N. Rothblum, and Ron~D. Rothblum.
\newblock {C}onstant-round interactive proofs for delegating computation.
\newblock In Daniel Wichs and Yishay Mansour, editors, {\em Proceedings of the
  48th Annual {ACM} {SIGACT} Symposium on Theory of Computing, {STOC} 2016,
  Cambridge, MA, USA, June 18-21, 2016}, pages 49--62. {ACM}, 2016.

\bibitem{RR20}
Noga Ron{-}Zewi and Ron~D. Rothblum.
\newblock {L}ocal {P}roofs {A}pproaching the {W}itness {L}ength [{E}xtended
  {A}bstract].
\newblock In {\em 61st {IEEE} Annual Symposium on Foundations of Computer
  Science, {FOCS} 2020, Durham, NC, USA, November 16-19, 2020}, pages 846--857.
  {IEEE}, 2020.

\bibitem{ethSTARK}
StarkWare.
\newblock eth{S}{T}{A}{R}{K} {D}ocumentation.
\newblock {\em {IACR} Cryptol. ePrint Arch.}, page 582, 2021.
\newblock URL: \url{https://eprint.iacr.org/2021/582}.

\bibitem{Sti93}
Henning Stichtenoth.
\newblock {\em {A}lgebraic function fields and codes}.
\newblock Universitext. Springer, 1993.

\bibitem{S08}
{Henning} Stichtenoth.
\newblock {\em Algebraic {F}unction {F}ields and {C}odes}.
\newblock Springer Publishing Company, Incorporated, 2nd edition, 2008.

\bibitem{TT14}
Saeed {Tafazolian} and Fernando {Torres}.
\newblock On the curve $y^n = x^m + x$ over finite fields.
\newblock {\em Journal of Number Theory}, 145:51--66, 2014.
\newblock \href {https://doi.org/10.1016/j.jnt.2014.05.019}
  {\path{doi:10.1016/j.jnt.2014.05.019}}.

\bibitem{TVZ82}
M.~A. Tsfasman, S.~G. Vl\u{a}du\c{t}, and Th. Zink.
\newblock Modular curves, {S}himura curves, and {G}oppa codes, better than
  {V}arshamov-{G}ilbert bound.
\newblock {\em Math. Nachr.}, 109:21--28, 1982.
\newblock \href {https://doi.org/10.1002/mana.19821090103}
  {\path{doi:10.1002/mana.19821090103}}.

\bibitem{TNV07}
Michael Tsfasman, Serge Vladut, and Dmitry Nogin.
\newblock {\em {A}lgebraic {G}eometric {C}odes: {B}asic {N}otions}.
\newblock American Mathematical Society, USA, 2007.

\bibitem{ZXZS20}
Jiaheng {Zhang}, Tiancheng {Xie}, Yupeng {Zhang}, and Dawn {Song}.
\newblock {T}ransparent {P}olynomial {D}elegation and {I}ts {A}pplications to
  {Z}ero {K}nowledge {P}roof.
\newblock In {\em 2020 {IEEE} Symposium on Security and Privacy, {SP} 2020, San
  Francisco, CA, USA, May 18-21, 2020}, pages 859--876. {IEEE}, 2020.
\newblock \href {https://doi.org/10.1109/SP40000.2020.00052}
  {\path{doi:10.1109/SP40000.2020.00052}}.

\end{thebibliography}

\appendix

\section{Proof of Proposition \ref{prop:BKS18-eta}}\label{app:prop-ldc}

Proposition \ref{prop:BKS18-eta} is a weighted version of \cite[Theorem~4.5]{BKS18}. We only highlight the changes to be made in the proof of \cite[Theorem~4.5]{BKS18}.

For $z \in \F$ and $(v_0,\dots,v_{l-1}) \in V^l$, let us set $\displaystyle{v_z \mydef \sum_{i = 0}^{l - 1} z^i v_i}$. Rewriting the proof of Theorem 4.5 \cite{BKS18} with
\[A \mydef\left\{z\in \F \mid \omega_\eta\left(u_z, V\right)>1-\delta \right\}\]
provides $v_0, \ldots, v_{l-1} \in V$ and a set
\[C\mydef \left\{z \in \F \mid \omega_\eta\left(u_z,v_z \right) > 1-\delta\right\} \subset A \]
with cardinality $|C| > \frac{l-1}{\epsilon}$. Let us set $T\mydef \{P \in \calP \mid \rest{u_i}{T}=\rest{v_i}{T} \text{ for all } i\}$. Therefore
\begin{align*}
	1-\delta &< \frac{1}{|C|} \sum_{z \in C} \omega_\eta\left(u_z, v_z\right)\\
	& =\frac{1}{|C|\times |\calP|} \sum_{z \in C} \sum_{P \in \calP} \eta(P) \mathds{1}_{u_z(P) = v_z(P)}\\
	& =\frac{1}{|\calP|} \sum_{P \in \calP} \eta(P) \frac{1}{|C|} \sum_{z \in C}  \mathds{1}_{u_z(P) = v_z(P)}
\end{align*}
Notice that if there exists $i \in\range{l}$ such that $u_i$ which does not coincide with $v_i$, the number of $z \in \F$ such that $u_z(P)= v_z(P)$ is at most $l-1$. Then 

\begin{align*}
	1-\delta & \leq \frac{1}{|\calP|} \sum_{P \in T} \eta(P) + \frac{1}{|\calP|} \sum_{P \in C\setminus T} \eta(P) \frac{l-1}{|C|}\\
	& \leq \frac{1}{|\calP|} \sum_{P \in T} \eta(P) + \epsilon,
\end{align*}
which gives the first item of the proposition.

\section{Properties of the genera of the curves in the Hermitian tower}\label{app:gen_HT}

To estimate the parameters of the foldable codes we define along the Hermitian tower, we need to handle the genera of the curves in this tower. From the formulae \eqref{eq:genus_herm}, we deduce a bound on the genus of the curve $\calX_i$ for small $i$ (Proposition \ref{prop:maj_genus_herm}) and the asymptotic behaviour of the ratio of $g_i$ by $q^{i+2}$ for $i=q^{\epsilon}$ when $q$ goes to infinity (Lemma \ref{lem:gen_sympt}).

\begin{proposition}\label{prop:maj_genus_herm}
	For $i \geq 1$, we have
	\[g_i \leq \dfrac{q^{i+1}}{2} \sum\limits_{k=1}^i \binom{i}{k} \dfrac{1}{q^{k-1}} \leq \dfrac{iq^{i+1}}{2} \sum\limits_{k=1}^i \left(\dfrac{i}{q}\right)^{k-1} \leq \dfrac{i}{2}q^{i+1} + \dfrac{i(i-1)}{2}q^i,\]
	the last inequality holding only if $2(i-1) < q$.
\end{proposition}

\begin{proof}
	Starting from the second formula of \eqref{eq:genus_herm}, we can write
	
	\begin{align*} g_i  &= \dfrac{1}{2} \cd \left(q^{i+1}\sum\limits_{k=1}^i \left(1+\frac{1}{q}\right)^{k-1} +1 -(1+q)^i\right)\\ 
    &\leq \dfrac{q^{i+1}}{2} \cd q \cd ((1+1/q)^i-1) \\ &= \dfrac{q^{i+1}}{2} \cd \sum\limits_{k=1}^i \binom{i}{k} \dfrac{1}{q^{k-1}},
    \end{align*}
	using that the term outside the geometric sum is non-positive. Note that if $k \geq 2$, then we can bound the binomial coefficients as follows
	\[\binom{i}{k} = \dfrac{i(i-1) \cdots (i-k+1)}{k(k-1) \cdots 2} \leq \dfrac{i(i-1)^{k-1}}{2},\] 
	as the denominator is greater than $2$ and the factors $i-1,i-2,...,i-k+1$ are all lesser than $i-1$. Factoring and using this upper bound on the binomial coefficients, we get
	\[  g_i \leq \dfrac{q^{i+1}}{2} \cd \left( i + \frac{i}{2} \sum\limits_{k=2}^i \left(\dfrac{i-1}{q}\right)^{k-1}\right) = \dfrac{iq^{i+1}}{2} \left( 1 + \frac{1}{2} \cd \left(\frac{i-1}{q}\right) \cd \sum\limits_{k=2}^i \left(\dfrac{i-1}{q}\right)^{k-2}\right).\]
	Assuming that $2(i-1) < q$, we can bound the sum  $\sum\limits_{k=2}^i \left(\dfrac{i-1}{q}\right)^{k-2}$ by $2$, which concludes the proof.
\end{proof}

\begin{lemma}\label{lem:gen_sympt}
	Fix $\epsilon \in (0,1)$ and set $i = q^{\epsilon}$. Then
	\[\frac{g_i}{q^{i+2}} \underset{q \rightarrow \infty}{\sim} \frac{1}{2q^{1-\epsilon}}.\]
\end{lemma}

\begin{proof}
	From the first formula of \eqref{eq:genus_herm}, we get
	\[\frac{2g{\ii}}{q^{\ii+2}}=\left(1-\frac{1}{q^2}\right)\left[\left(1+\frac{1}{q}\right)^{q^\epsilon}-1\right]+\frac{1}{q^{\ii+2}}-\frac{1}{q^2}\]
	Let us examine the asymptotic behaviour of $\left(1+\frac{1}{q}\right)^{q^\epsilon}$ when $q$ goes to infinity. Set $h = q^{-1}$.
	
	\[\left(1+\frac{1}{q}\right)^{q^\epsilon}=\exp\left(h^{1-\epsilon}\cdot\frac{\log(1+h)}{h}\right)=\exp\left(h^{1-\epsilon}\left(1-\frac{h}{2}+o(h)\right)\right)=1+h^{1-\epsilon}+o(h^{1-\epsilon})\]
	Therefore, we have
	\[\left(1+\frac{1}{q}\right)^{q^\epsilon}-1 \sim \frac{1}{q^{1-\epsilon}}.\]
\end{proof}

\end{document}